%% file: paper.tex
\documentclass[runningheads]{llncs}
\usepackage[T1]{fontenc}
\usepackage{graphicx}
\usepackage{amsmath,amssymb,mathtools}
\usepackage{nicefrac}
\usepackage{microtype}
\usepackage{longtable}

\usepackage{adjustbox}

\usepackage{orcidlink}

\usepackage{tikz}
\usetikzlibrary{
	arrows, %
	automata, %
	positioning, %
	calc, %
    patterns.meta, %
}

\usepackage{xcolor}
\newcommand{\ifempty}[3]{\ifthenelse{\equal{#1}{}}{#2}{#3}}

\colorlet{colorProbability}{teal}
\colorlet{colorAction}{magenta}
\colorlet{colorReward}{green}

\tikzstyle{labelnode}=[font=\footnotesize]
\tikzstyle{dist}=[circle,  inner sep=0.8pt, solid, draw,fill]
\tikzstyle{mdp}=[every label/.style={labelnode}]
\tikzset{ps/.style={draw,circle,inner sep=0pt,text centered,minimum size=6mm,text width=#1},ps/.default=3mm}
\tikzset{pswide/.style={draw,ellipse,minimum size=6mm,text centered,text width=#1},pswide/.default=7mm}
\tikzstyle{ptrans}=[->,semithick,labelnode]
\tikzstyle{every initial by arrow}=[inner sep=0pt] %
\tikzset{init/.style={initial #1, initial text={}, initial distance=4mm},init/.default=left}
\tikzset{anode/.style={circle, fill=black, inner sep=1pt, minimum size=2pt}}

\definecolor{myyellow}{RGB}{227, 222, 91}
\definecolor{myblue}{RGB}{103, 182, 219}
\definecolor{myorange}{RGB}{250, 195, 162}
\tikzset{cmdpyellow/.style={pattern={Lines[angle=45,distance={3pt/sqrt(2)}]}, pattern color=myyellow}}
\tikzset{cmdpblue/.style={pattern={Hatch}, pattern color=myblue}}
\tikzset{cmdporange/.style={pattern={Dots[distance=2pt]}, pattern color=myorange}}

\tikzset{cinit/.style={pattern={Lines[angle=45,distance={3pt/sqrt(2)}]},pattern color=myblue}}
\tikzset{cinite/.style={myblue}}

\newcommand{\nodebg}[1]{{\setlength{\fboxsep}{0.6pt}\colorbox{white}{#1}}}

\usepackage{todonotes}

\usepackage{multirow}
\usepackage{tabularx}
\usepackage{booktabs}
\usepackage{float}

\usepackage[vlined,linesnumbered]{algorithm2e}
\SetAlCapSkip{2ex plus 0.5ex minus 1ex}
\SetEndCharOfAlgoLine{}
\SetArgSty{textrm}
\SetKwInOut{Input}{Input}\SetKwInOut{Output}{Output}
\definecolor{commentcolor}{RGB}{60,114,26}

\SetCommentSty{commentstyle}
\SetKw{Break}{break}
\SetKwProg{Type}{type}{}{end}
\SetKwProg{Function}{function}{}{end}
\SetKwFor{ForEach}{foreach}{do}{}
\SetKwFor{WhileTrue}{repeat}{}{end}
\SetKw{Continue}{continue}
\SetKwRepeat{DoWhile}{repeat}{while}
\SetKwRepeat{DoUntil}{repeat}{until}

\usepackage{caption}
\usepackage{subcaption}
\usepackage{hyperref} 
\usepackage[nameinlink,noabbrev,capitalise]{cleveref}

\usepackage[framemethod=TikZ]{mdframed}
\mdfdefinestyle{theoremstyle}{%
linecolor=black,linewidth=1pt,%
innertopmargin=\topskip,
} 
\mdtheorem[style=theoremstyle]{problemstatement}{Problem}
\Crefname{problemstatement}{Problem}{Problems}
\crefname{problemstatement}{Problem}{Problems}
\crefname{algorithm}{algorithm}{algorithms}
\Crefname{algorithm}{Algorithm}{Algorithms}

\newcommand{\RR}{\mathbb{R}}

\newcommand{\dist}{\ensuremath{\mathit{Dist}}}
\newcommand{\supp}{\ensuremath{\mathit{Supp}}}
\newcommand{\set}[1]{\ensuremath{\{ #1 \}}}
\newcommand{\bigset}[1]{\ensuremath{\big\{ #1 \big\}}}

\newcommand*{\bigmid}{\mathrel{\bigm|}}

\newcommand{\iverson}[1]{\ensuremath{\big[\, #1 \,\big]}}

\newcommand{\mdp}{\ensuremath{\mathcal{M}}}
\newcommand{\mdpexa}{\ensuremath{\mdp_1}}
\newcommand{\mdpexb}{\ensuremath{\mdp_2}}
\newcommand{\act}{\ensuremath{\mathit{Act}}}
\newcommand{\sinit}{\ensuremath{{s}^{I}}}
\newcommand{\cinit}{\ensuremath{{C}^{I}}}
\newcommand{\Pbf}{\ensuremath{\mathbf{P}}}
\newcommand{\goal}{\ensuremath{G}}
\newcommand{\othergoal}{\ensuremath{T}}
\newcommand{\cond}{\ensuremath{E}}
\newcommand{\paths}[1][\mdp]{\ensuremath{\mathit{Paths}_{#1}}}
\newcommand{\last}{\ensuremath{\mathit{last}}}
\newcommand{\pr}[2]{\ensuremath{\mathrm{Pr}_{#1}^{#2}}}
\newcommand{\prms}{\pr{\mdp}{\sigma}}
\newcommand{\prss}{\pr{s}{\sigma}}
\newcommand{\prmmax}{\pr{\mdp}{\max}}
\newcommand{\prsmax}{\pr{s}{\max}}
\newcommand{\prmmin}{\pr{\mdp}{\min}}
\newcommand{\prsmin}{\pr{s}{\min}}
\newcommand{\ex}[2]{\ensuremath{\mathrm{Ex}_{#1}^{#2}}}
\newcommand{\exms}{\ex{\mdp}{\sigma}}
\newcommand{\exmmax}{\ex{\mdp}{\max}}
\newcommand{\exsmax}{\ex{s}{\max}}

\newcommand{\rew}{\ensuremath{\mathcal{R}}}
\newcommand{\pols}[1][]{\ensuremath{\Sigma_{#1}}}
\newcommand{\reach}[1][]{\ensuremath{\mathit{Reach}_{#1}}}
\newcommand{\exits}[1][]{\ensuremath{\mathit{Exits}}}
\newcommand{\sizeof}[1]{\ensuremath{\lVert #1 \rVert}}
\newcommand{\val}{V}

\newcommand{\tool}[1]{\textsc{#1}\xspace}
\newcommand{\storm}{\tool{Storm}}

\newlength\myheight
\newlength\mydepth
\settototalheight\myheight{Xygp}
\newcommand*\inlinegraphics[1]{%
  \settototalheight\myheight{Xygp}%
  \settodepth\mydepth{Xygp}%
  \raisebox{-0.5\mydepth}{\includegraphics[height=\myheight]{#1}}%
}

\newcommand{\refarxiv}[2]{#2} %

\newcommand{\onlyarxiv}[1]{\refarxiv{}{#1}}

\renewcommand{\paragraph}[1]{\smallskip\noindent\emph{#1}~}

\begin{document}
\title{Fast Computation of Conditional Probabilities in MDPs and Markov Chain Families}
\titlerunning{Fast Computation of Conditional Probabilities}
\author{%
Milan Češka\inst{1}\orcidlink{0000-0002-0300-9727} \and
Sebastian Junges\inst{2}\orcidlink{0000-0003-0978-8466} \and
Luko van der Maas\inst{2}\orcidlink{0009-0007-3915-6191} \and\\
Filip Macák\inst{1}\orcidlink{0009-0004-4277-2751} \and
Tim Quatmann\inst{3}\orcidlink{0000-0002-2843-5511}
}
\authorrunning{Češka et al.}
\institute{%
Brno University of Techntology, Brno, Czech Republic \and
Radboud University, Nijmegen, the Netherlands \and
RWTH Aachen University, Aachen, Germany
}
\maketitle              %
\begin{abstract}
Computing optimal conditional reachability probabilities in Markov decision processes (MDPs) is tractable by a reduction to reachability probabilities. Yet, this reduction yields cyclic, challenging MDPs that are often notoriously hard to solve.
We present an alternative, practically efficient method to compute optimal conditional reachabilities.
This new method is numerically stable, can decide the threshold problem in linear time on acyclic MDPs, and yields performance comparable to standard reachability queries. 
We also integrate the method in an abstraction-refinement framework to analyse millions of Markov chains at once. We demonstrate the efficacy of the new methods on benchmarks from Bayesian network analysis, probabilistic programs, and runtime monitoring and show speed-ups up to multiple orders of magnitude.
\end{abstract}
\section{Introduction}
Markov chains (MCs) are a standard model for stochastic processes. The standard reachability query on Markov chains is, e.g., `What is the probability that one reaches an airport before the plane departs?' As a variation, we study `What is the probability that one reaches the airport, if you already know that the buses will have at least a ten-minute delay?'  Such a query is a typical example for \emph{conditional reachability probabilities}~\cite{DBLP:conf/tacas/AndresR08}. 
Intuitively, conditional reachability asks for reachability probabilities not over all paths from the initial state, but over paths that match the given evidence. It can be used to implicitly define an alternative initial state, that is obtained from observing a running system.
Conditional reachability probabilities have prominent applications: 
In runtime monitoring, one aims to compute the probability of an imminent threat conditioned on the previous observations by the monitor~\cite{DBLP:conf/rv/StollerBSGHSZ11}. 
In Bayesian networks and MCs derived from them~\cite{DBLP:conf/qest/SalmaniK20}, one computes the probability of an event like a break-in conditioned on evidence like a burglary alarm.
In probabilistic programs and their MC semantics, one aims to compute the probability of a return value conditioned on known inputs, e.g.,~\cite{DBLP:journals/pacmpl/HoltzenBM20}. 
These examples of conditional reachability probabilities can be efficiently computed~\cite{DBLP:conf/tacas/BaierKKM14} on reasonably-sized MCs using modern probabilistic model checking tools~\cite{DBLP:journals/sttt/HenselJKQV22,DBLP:conf/tacas/ForejtKNPQ11}. In this paper, we  bring this performance to Markov decision processes and to Markov chain families by a fresh perspective. 

\subsubsection*{Towards MDPs.} 
Many systems cannot be modelled concisely as a Markov chain: 
Not all behavior can be captured probabilistically; probabilities may not be known, or the stochastic process must be abstracted by a smaller state space. 
These observations are standard motivations to add nondeterminism to the model formalism and to thus consider \emph{Markov Decision Processes} (MDPs). The ability to compute optimal (i.e., maximal/minimal) conditional reachability probabilities in MDPs then allows us to bound the conditional reachability probability in models for runtime monitors with nondeterministic agents~\cite{DBLP:conf/cav/JungesTS20} or with learned probabilities~\cite{luko-aamas}, Bayesian networks with interval-valued probabilities~\cite{DBLP:journals/jair/SalmaniK23}, or probabilistic programs under abstractions~\cite{DBLP:journals/pacmpl/ChoGH25} or with explicit nondeterminism~\cite{guertlerkaminski}.

\paragraph{A Bayesian perspective.}
Computing conditional reachability probabilities in MCs is straightforward by applying Bayes' theorem: The probability to reach the goal states $\goal$ conditioned on also reaching evidence states $\cond$ can be expressed as the fraction of two reachability probabilities: The probability to reach $\goal$ and $\cond$ divided by the probability to reach $\cond$, i.e.,
\[ \Pr(\lozenge \goal \mid \lozenge \cond) ~=~ \frac{\Pr(\lozenge \goal \cap \lozenge \cond)}{\Pr(\lozenge \cond)}. \]
The upper probability can be computed on the product of a small DFA and the original Markov chain.
In MDPs, however, we optimize over the policies that resolve the nondeterminism, which prevents calculating the numerator and denominator independently:
\[ \max_\sigma \frac{\Pr_\sigma(\lozenge \goal \cap \lozenge \cond)}{\Pr_\sigma(\lozenge \cond)} ~\neq~ \frac{\max_\sigma \Pr_\sigma(\lozenge \goal \cap \lozenge \cond)}{\max_\sigma \Pr_\sigma(\lozenge \cond)}.  \]

\paragraph{The restart method.} Baier \emph{et al.}~\cite{DBLP:conf/tacas/BaierKKM14} provide a mathematically elegant solution to this challenge, basically mimicking the key idea of rejection sampling. Any probability mass that never satisfies the condition $\cond$ cannot be reached must be redistributed (i.e., restarted). In terms of the MDP, this means that for all states from which $\cond$ cannot be reached, the transitions are rewired to the initial state. In these new (restart) MDPs, one simply computes the probability to eventually reach the goal! The downside is that this method introduces a lot of loops, in particular even when the original MDP was acyclic~\cite{DBLP:conf/sefm/MarckerB0K17}. MDPs with cycles back to the initial states are among the most challenging for methods based on value iteration~\cite{DBLP:journals/tcs/HaddadM18}. Indeed, some(!) restart MDPs are known as very challenging benchmarks of standard reachability queries~\cite{DBLP:conf/sefm/MarckerB0K17,DBLP:conf/tacas/HartmannsJQW23}.

\paragraph{A multiobjective perspective.}
Intuitively, we consider conditional reachability queries as a variation of a bi-objective model checking problem~\cite{DBLP:journals/lmcs/EtessamiKVY08}: Maximize the probability to reach $\goal$ and $\cond$, while minimizing the probability to reach~$\cond$. Consider the decision problem whether $\max_\sigma \Pr_\sigma(\lozenge \goal \mid \lozenge \cond) \geq \lambda$. 
Using Bayes' law, some reordering, and some side-conditions detailed below, this can be seen~as \[ \textstyle  \max_\sigma \Pr_\sigma(\lozenge \goal \cap \lozenge \cond) - \lambda \Pr_\sigma(\lozenge \cond) \geq 0. \]
This query can be efficiently computed (in practice and theory) using one so-called \emph{total reward query}~\cite{DBLP:conf/tacas/ForejtKNPQ11,DBLP:conf/tacas/QuatmannK21} and is a simple instance of a \emph{relational reachability property}~\cite{DBLP:conf/cav/GerlachWABJ25}. If the original MDP is acyclic this is possible in linear time with a numerically stable method. 

\paragraph{Optimal conditional probabilities.}
The decision method above gives rise to a method to approximate the optimal conditional reachability probability, using a standard bisection method. We accelerate this method using a smarter heuristic. Further, note that we also support computing the optimal conditional reachability probability. The main argument for the bisection to terminate is that the solution can be represented as a reachability probability~\cite{DBLP:conf/tacas/BaierKKM14} and that these must be a rational number representable by a bounded number of bits~\cite{DBLP:conf/spin/ChatterjeeH08}. We therefore adapt a routine that bisects this set of bounded rationals~\cite{DBLP:journals/fmsd/MathurBCSV20} and extend it with termination criteria based on the computed optimal policies.

\subsubsection*{Towards Markov chain families.}
Beyond MDPs, we are interested in analysing Markov chain families, i.e., large sets of Markov chains. Our interest is twofold: 
(1)~Markov chain families can capture uncertainty about the topology of a system and can capture systems that can be configured in various ways. Such families occur, e.g., when analysing (purely) probabilistic programs with non-deterministic inputs~\cite{DBLP:conf/tacas/BatzCJKKM23,DBLP:journals/pacmpl/ChoGH25}.
(2)~Given sets of states $\cond_1,\dots \cond_n$, we are interested whether there is some evidence $\cond_i$ such that $\Pr(\lozenge \goal \mid \lozenge \cond_i) \geq \lambda$. This question is natural in runtime monitoring~\cite{DBLP:conf/tacas/BadingsVJSJ24,DBLP:conf/atva/MaasJ25}, where evidence itself may be uncertain and programs that encode a secret input (e.g., ciphers, see Section~\ref{sec:experiments}), where we want to ensure a secret cannot be reasonably guessed from any possible evidence. 

\paragraph{Embedding into abstraction-refinement.}
We use a standard representation of Markov chain families using \emph{colored MDPs}, where every color-consistent policy of the MDP represents a Markov chain~\cite{DBLP:journals/jair/AndriushchenkoCMJK25}. 
To solve colored MDPs, we use abstraction-refinement~\cite{DBLP:journals/jair/AndriushchenkoCMJK25}, the state-of-the-art method for solving standard reachability queries in colored MDPs. However, the shift to conditional probabilities brings various challenges, that prevent simply applying previous methods: The optimal policies in (standard) MDPs are not memoryless and we only get a value for the initial state. Our experiments show that, nevertheless, abstraction-refinement can bound conditional probabilities on millions of MCs at once.

\subsubsection*{Related work.}
The previous state-of-the-art is the restart method, outlined above, which was experimented with in~\cite{DBLP:conf/sefm/MarckerB0K17}. The approach in this paper is completely different and rooted in multi-objective MDP model checking. To the best of our knowledge, the first MDP model checking algorithm for conditional reachability probabilities is presented in \cite{DBLP:conf/tacas/AndresR08}: For acyclic MDPs, it suggests a method that computes numerator and denominator under each policy, along with some pruning that is broadly related to multiobjective model checking. However, the approach remains exponential and highly inefficient. 

Our approach of tackling the threshold problem and using that as a subroutine also appears in work on conditional expected rewards~\cite{DBLP:conf/tacas/Baier0KW17}. Deciding the threshold problem for conditional expected rewards is already PSPACE-complete for acyclic MDPs, in stark contrast to our linear time algorithm. Whereas~\cite{DBLP:conf/tacas/Baier0KW17} focusses on establishing this and other complexity bounds, we present a concise algorithm with a focus on (practical) performance: We use value iteration over LPs, we efficiently extract policies, and use efficient bisection methods. %

The recent work~\cite{DBLP:journals/pacmpl/ChoGH25} develops novel probabilistic programming languages (PPLs) to solve various optimization problems over probabilistic inference. It explicitly remarks that computing conditional probabilities on MDPs directly is ill-supported. With this paper, we make reasoning about the operational semantics of PPLs with conditioning more feasible. 

\subsubsection*{Contributions.}
In summary, this paper yields a significant step forward in computing conditional probabilities on MDPs, building upon progress in (multiobjective) model checking from the last decade. In particular, we provide an empirically superior method for computing optimal conditional reachability probabilities in MDPs rooted in a reduction to a total reward query paired with an optimised bisection method.
The method is further equipped with a policy extraction technique that enables an effective integration into the abstraction refinement scheme to solve conditional reachability queries on colored MDPs. Our experiments on benchmarks from several application domains demonstrate speed-ups up to multiple orders of magnitude compared to the restart method~\cite{DBLP:conf/tacas/BaierKKM14}.

\section{Preliminaries}
The set of probability distributions over a finite set $S$ is $\dist(S) = \set{ \mu \colon S \to [0,1] \mid \sum_{s\in S} \mu(s) = 1 }$, 
$\supp(\mu) = \set{s \in S \mid \mu(s) > 0}$ is the support of $\mu \in \dist(S)$.

\begin{definition}
    A \emph{Markov decision process (MDP)} is a tuple $\mdp = (S,\act,\Pbf,\sinit)$, where $S$ and $\act$ are finite sets of states and actions, $\sinit$ is an initial state, and $\Pbf \colon S \times \act \rightharpoonup \dist(S)$ is a (partial) transition probability function such that
    $\act(s) = \set{ \alpha \in \act \mid \Pbf(s,\alpha) \text{ is defined }} \neq \emptyset$ for all $s \in S$.
\end{definition}
We fix an MDP $\mdp = (S,\act,\Pbf,\sinit)$.
$\sizeof{\mdp} = |\set{ (s,\alpha,s') \mid \Pbf(s,\alpha)(s') > 0}|$ denotes the size of $\mdp$. An MDP is \emph{acyclic} if all cycles in its transition graph are self-loops.
A \emph{path} of MDP $\mdp$ is an infinite sequence $\pi = s_0\xrightarrow{\alpha_0} s_1 \xrightarrow{\alpha_1} \dots$ with $\alpha_i \in \act(s_i)$ and $s_{i+1} \in \supp(\Pbf(s_i,\alpha_i))$ for all $i \ge 0$.
$\paths$ denotes the set of paths of $\mdp$.
A finite path $\hat\pi = s_0\xrightarrow{\alpha_0} \dots \xrightarrow{\alpha_{n-1}} s_n$ is a finite prefix of a path ending in state $\last(\hat\pi) = s_n$.
For a set of states $G \subseteq S$ we define the set of paths that eventually reach a state in $G$ as
\[
\lozenge \goal ~=~ \bigset{s_0\xrightarrow{\alpha_0} s_1 \xrightarrow{\alpha_1} \dots \in \paths \bigmid s_i \in \goal \text{ for some } i \ge 0}.
\]
A \emph{policy} $\sigma$ for $\mdp$ maps finite paths $\hat\pi$ of $\mdp$ to enabled actions $\sigma(\hat\pi) \in \act(\last(\hat\pi))$.
The set of all policies for $\mdp$ is denoted as $\pols[\mdp]$, when $\mdp$ is clear from context, we also write $\pols$.
A memoryless policy satisfies $\sigma(\hat\pi) = \sigma(\hat\pi')$ for all finite paths $\hat\pi$ and $\hat\pi'$ with $\last(\hat\pi) = \last(\hat\pi')$. 
As usual, we also write $\sigma\colon S \rightarrow \act$ for memoryless policies, mapping last states to actions.
MDP $\mdp$ and policy $\sigma$ for $\mdp$ induce a probability measure $\prms$ on measurable sets of paths defined via a standard cylinder set construction~\cite{BK08}.
For two measurable sets of paths $\Pi_1, \Pi_2 \subseteq \paths$, the conditional probability of $\Pi_1$ given $\Pi_2$ is defined as
\[
\prms(\Pi_1 \mid \Pi_2) ~=~ \nicefrac{\prms(\Pi_1 \cap \Pi_2)}{\prms(\Pi_2)}.
\]
The probability is undefined iff $\prms(\Pi_2) = 0$.
We define 
\[
\prmmax(\Pi_1 \mid \Pi_2) ~=~~ \sup \bigset{ \prms(\Pi_1 \mid \Pi_2) \bigmid \sigma \in \pols[\mdp] \text{ with } \prms(\Pi_2) \neq 0}.
\]
$\prmmin(\Pi_1 \mid \Pi_2)$ is defined similarly.
When $\mdp$ is clear from context and $s \in S$ is a state of $\mdp$, we also write $\prss$, $\prsmax$, and $\prsmin$ for the corresponding probability measures that arise from replacing the initial state of $\mdp$ by $s$. 

A \emph{reward function} $\rew \colon S \times \act \times S \to \RR$ for MDP $\mdp$ assigns a reward to transitions. The reward accumulated on a path 
$\pi = s_0\xrightarrow{\alpha_0} s_1 \xrightarrow{\alpha_1} \dots$ until reaching a set  $\othergoal \subseteq S$ is given by
$\rew_{\lozenge \othergoal}(\pi) = \sum_{0 \le i < \min \set{n \in \mathbb{N} \mid s_n \in \othergoal}} \rew(s_i, \alpha_i, s_{i+1})$.\footnote{In our setting, this sum always converges to a finite value. Generally, it might be infinite or undefined---due to possible mixtures of positive and negative rewards.}
The \emph{expected total reward} of $\mdp$ for $\rew_{\lozenge \othergoal}$ under a policy $\sigma \in \pols[\mdp]$ is the expected value of the accumulated reward: $
\exms(\rew_{\lozenge \othergoal}) ~=~ \int_{\pi} \rew_{\lozenge \othergoal}(\pi) \,d \prms(\pi)$.
As for $\prms$, we write $\exmmax$ and $\exsmax$ for the \emph{maximal} expected values (assuming initial state $s \in S$).
The values $\exsmax(\rew_{\lozenge \othergoal})$ for each $s \in S$ coincide with the least fixed point $\vec{x} \colon S \to \RR$ of the \emph{Bellman operator} $\Phi \colon (S \to \RR) \to (S \to \RR)$ given by
$
\Phi(\vec{x})(s) ~=~ \max_{\alpha \in \act(s)} \sum_{s' \in S} \Pbf(s,\alpha,s') \cdot \big(\vec{x}(s') + \rew(s,\alpha,s')\big)
$ for $s \notin \othergoal$ and $\Phi(\vec{x})(s) = 0$ for $s \in \othergoal$.
Common algorithms for computing $\exsmax(\rew)$ for each $s \in S$ are based on linear programming, policy iteration, and value iteration.
When mixtures of positive and negative rewards are present (as in our case), these algorithms commonly require collapsing end components as a preprocessing to ensure a \emph{unique} fixed point of the Bellman operator~\cite{DBLP:phd/us/Alfaro97}.

\section{Deciding Conditional Reachability Properties}
\label{sec:threshold}

This section presents our new approach for answering the following  problem.

\begin{problemstatement}[Conditional Reachability Threshold Problem]\label{prob:threshold}
\vspace{-1em}
Given an MDP $\mdp$ with sets of states $\goal,\cond \subseteq S$ and a threshold $\lambda \in [0,1]$, decide whether $\prmmax(\lozenge \goal \mid \lozenge \cond) \le \lambda$.
\end{problemstatement}
For the rest of the section, we fix an MDP $\mdp$ with state sets $\goal,\cond \subseteq S$.
We assume that the conditional probability $\prmmax(\lozenge \goal \mid \lozenge \cond)$ is defined, i.e., $\prms(\lozenge \cond) > 0$ for some $\sigma \in \pols[\mdp]$.
This assumption can be checked using standard qualitative MDP model checking algorithms~\cite{BK08}.
The restart method outlined in~\cite{DBLP:conf/tacas/BaierKKM14} provides a polynomial time algorithm for computing $\prmmax(\lozenge \goal \mid \lozenge \cond)$ and thus for solving \Cref{prob:threshold}.
Roughly, the approach constructs a new MDP $\mdp'$ from $\mdp$, where transitions leading to states that do not reach the condition $\cond$ are rewired to the initial state.
The conditional reachability probability $\prmmax(\lozenge \goal \mid \lozenge \cond) \le \lambda$ in $\mdp$ coincides with the (unconditional) reachability probability $\pr{\mdp'}{\max}(\lozenge \goal)$ in $\mdp'$.
The latter can be computed using standard MDP model checking methods.
Since $\mdp'$ has approximately the same size as $\mdp$, the linear programming approach yields a polynomial time algorithm for \Cref{prob:threshold}.

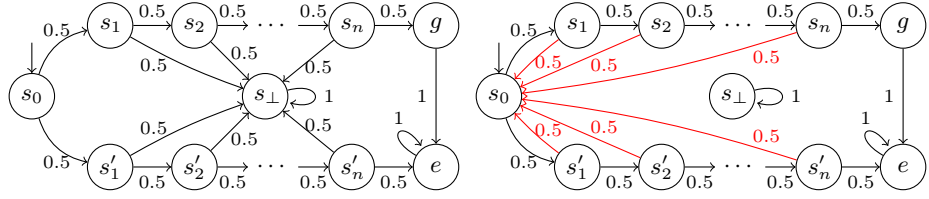
\begin{figure}[t]\centering
\begin{tikzpicture}[mdp]
\node[ps, init=above] (s0) {$s_0$};
\node[ps,above  right=0.5 and 0.6 of s0] (s1) {$s_{1}$};
\node[ps, below right=0.5 and 0.6 of s0] (s1p) {$s_{1}'$};
\node[ps, right=0.5 of s1] (s2) {$s_2$};
\node[ right=0.4 of s2] (d) {$\cdots$};
\node[ps, right=0.4 of d] (sn) {$s_n$};
\node[ps, right=0.5 of sn] (g) {$g$};
\node[ps, right=0.5 of s1p] (s2p) {$s_2'$};
\node[ right=0.4 of s2p] (dp) {$\cdots$};
\node[ps, right=0.4 of dp] (snp) {$s_n'$};
\node[ps, right=0.5 of snp] (h) {$e$};
\node[ps, below right=0.5 and 0.5 of s2] (b) {$s_\bot$};
\path[->,font=\scriptsize]
(s0) edge[bend left] node[above]{0.5} (s1)
(s0) edge[bend right] node[below]{0.5} (s1p)
(s1) edge node[above]{0.5}(s2)
(s2) edge node[above]{0.5}(d)
(d) edge node[above]{0.5}(sn)
(sn) edge node[above]{0.5}(g)
(g) edge[] node[left]{1}(h)
(s1p) edge node[below]{0.5}(s2p)
(s2p) edge node[below]{0.5}(dp)
(dp) edge node[below]{0.5}(snp)
(snp) edge node[below]{0.5}(h)
(h) edge[out=150,in=125,loop] node[above]{1}(h)
;
\path[->,font=\scriptsize]
(s1) edge[bend right=5] node[pos=0.2,below]{0.5}(b)
(s1p) edge[bend left=5] node[pos=0.2,above]{0.5}(b)
(s2) edge[] node[pos=0.3,right]{0.5}(b)
(s2p) edge[] node[pos=0.3,right]{0.5}(b)
(sn) edge[] node[pos=0.3,below]{0.5}(b)
(snp) edge[] node[pos=0.3,above]{0.5}(b)
(b) edge[loop right] node[right]{1}(b)
;
\end{tikzpicture}~
\begin{tikzpicture}[mdp]
\node[ps, init=above] (s0) {$s_0$};
\node[ps,above  right=0.5 and 0.6 of s0] (s1) {$s_{1}$};
\node[ps, below right=0.5 and 0.6 of s0] (s1p) {$s_{1}'$};
\node[ps, right=0.5 of s1] (s2) {$s_2$};
\node[ right=0.4 of s2] (d) {$\cdots$};
\node[ps, right=0.4 of d] (sn) {$s_n$};
\node[ps, right=0.5 of sn] (g) {$g$};
\node[ps, right=0.5 of s1p] (s2p) {$s_2'$};
\node[ right=0.4 of s2p] (dp) {$\cdots$};
\node[ps, right=0.4 of dp] (snp) {$s_n'$};
\node[ps, right=0.5 of snp] (h) {$e$};
\node[ps, below right=0.5 and 0.5 of s2] (b) {$s_\bot$};

\path[->,font=\scriptsize]
(s0) edge[bend left] node[above]{0.5} (s1)
(s0) edge[bend right] node[below]{0.5} (s1p)
(s1) edge node[above]{0.5}(s2)
(s2) edge node[above]{0.5}(d)
(d) edge node[above]{0.5}(sn)
(sn) edge node[above]{0.5}(g)
(g) edge[] node[left]{1}(h)
(s1p) edge node[below]{0.5}(s2p)
(s2p) edge node[below]{0.5}(dp)
(dp) edge node[below]{0.5}(snp)
(snp) edge node[below]{0.5}(h)
(h) edge[out=150,in=125,loop] node[above]{1}(h)
(b) edge[loop right] node[right]{1}(b)
;
\path[->,red,font=\scriptsize]
(s1) edge[bend right=5] node[pos=0.2,below]{0.5}(s0)
(s1p) edge[bend left=5] node[pos=0.2,above]{0.5}(s0)
(s2) edge[] node[pos=0.3,below]{0.5}(s0)
(s2p) edge[] node[pos=0.3,above]{0.5}(s0)
(sn) edge[bend left=5] node[pos=0.1,below]{0.5}(s0)
(snp) edge[bend right=5] node[pos=0.1,above]{0.5}(s0)
;
\end{tikzpicture}
\caption{Example MDPs $\mdpexa$ (left) and $\mdpexa'$ (right).}
\label{fig:restart}
\end{figure}
\begin{example}
\Cref{fig:restart} shows an \emph{acyclic} MDP $\mdpexa$ and the corresponding MDP $\mdpexa'$ that arises from applying the restart method for the conditional probability $\prmmax(\lozenge \set{g} \mid \lozenge \set{e})$.
In $\mdpexa'$, transitions that are rewired to the initial state are depicted in red. In $\mdpexa$, these transitions lead to $s_\bot$ which cannot reach $e$.
We have $\prmmax(\lozenge \set{g} \mid \lozenge \set{e}) ~=~ \pr{\mdpexa'}{\max}(\lozenge \set{g}) = 0.5$.
However, computing the reachability probability in $\mdpexa'$ is challenging due to its cyclic structure.
A similar MDP is used in \cite[Fig. 3]{DBLP:journals/tcs/HaddadM18} to illustrate convergence issues of value iteration.
\end{example}

We present an alternative approach for answering \Cref{prob:threshold} that does not introduce loops into the MDP structure.
Let $\pols[\mdp,\cond] = \set{ \sigma \in \pols \mid \prms(\lozenge \cond) > 0}$ denote the set of policies that reach $\cond$ with positive probability.
Our procedure is based on the following key observation (proof in \refarxiv{\cite{arxiv}}{\cref{proofs:threshold}}).

\begin{theorem}\label{thm:threshold}
For $\lambda \in [0,1]$ and ${\sim} \in \set{<,\le, =, \ge,>}$:
\[
\prmmax(\lozenge \goal \mid \lozenge \cond) ~\sim~ \lambda
\quad\textrm{iff}\quad
\max_{\sigma \in \pols[\mdp,\cond]} \Big(\prms(\lozenge \goal \cap \lozenge \cond) - \lambda \cdot \prms(\lozenge \cond)\Big) ~\sim~0.
\]
\end{theorem}
Using \Cref{thm:threshold}, our approach is to answer \Cref{prob:threshold} by computing the value 
\[
\val(\lambda) ~\coloneq~\max_{\sigma \in \pols[\mdp,\cond]} \val^\sigma(\lambda) 
\quad \text{with} \quad
\val^\sigma(\lambda) ~\coloneq~
\Big(\prms(\lozenge \goal \cap \lozenge \cond) - \lambda \cdot \prms(\lozenge \cond)\Big)
\]
and comparing it against 0. The computation of $\val(\lambda)$ is done in two steps. (1)~Restate $\val^\sigma(\lambda)$ as an expected total reward. (2)~Construct an equivalent MDP $\mdp'$ for which the set $E$ is reached with positive probability under all policies (i.e., $\pols[\mdp',\cond] = \pols[\mdp']$).
After step~(2), standard MDP algorithms for expected total rewards can be employed to decide $\val(\lambda) \leq 0$.

\subsubsection*{Step 1: Reward Characterization.}
The value $\val(\lambda)$ above is characterized using an expected total reward for the reward function $\rew^\lambda$ for $\lambda \in [0,1]$, given as:
\[
\rew^\lambda(s,\alpha,s') ~=~ \begin{cases}
\pr{s'}{\max}(\lozenge \goal) - \lambda & \text{if } s' \in \cond \\
\pr{s'}{\max}(\lozenge \cond) \cdot (1-\lambda) &\text{if } s' \in \goal \setminus \cond \\
0 & \text{otherwise}
\end{cases}
\]
Intuitively, once a state in $\cond$ is reached, a policy $\sigma$ that maximizes $\val^\sigma(\lambda)$ as above must maximize the probability to also reach $\goal$.
Similarly, once a state in $\goal$ is reached, such a policy maximizes the probability to also reach $\cond$.

\begin{lemma}\label{lem:valueasrew}
    $\displaystyle \val(\lambda) = \max_{\sigma \in \pols[\mdp,\cond]} \exms \big(\rew^\lambda_{\lozenge \goal \cup \cond}\big)$
\end{lemma}
\begin{proof}[sketch]
  We can write $\val^\sigma(\lambda)$ as an expected value of a function $f^\lambda$ with
  \[
  f^\lambda(\pi) ~=~ \iverson{ \pi \in \lozenge \goal \cap \lozenge \cond } - \lambda \cdot \iverson{\pi \in \lozenge \cond}.
  \]
  The (maximal) expected values of $f^\lambda$ and $\rew^\lambda_{\lozenge \goal \cup \cond}$ coincide.
\end{proof}
We now exemplify how we apply the reward characterization to solve \Cref{prob:threshold}.

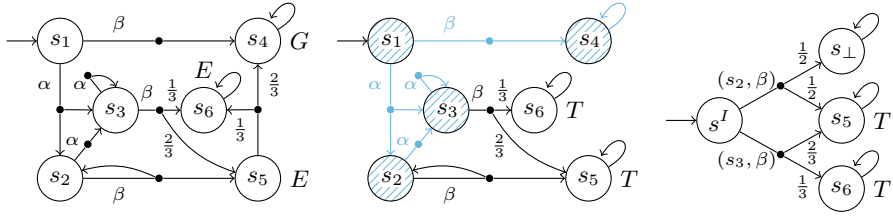
\begin{figure}[t]
    \centering
    \begin{tikzpicture}[mdp]
        \node[ps, init=left] (s1) {$s_1$};
        \node[ps, below=1.2 of s1] (s2) {$s_2$};
        \node[ps, below right=0.48 and 0.3 of s1] (s3) {$s_3$};
        \node[ps, right=2.0 of s1, label={right:$\goal$}] (s4) {$s_4$};
        \node[ps, right=2.0 of s2, label={right:$\cond$}] (s5) {$s_5$};
        \node[ps, right=0.55 of s3, label={above:$\cond$}] (s6) {$s_6$};
        \node[anode] (a12) at ($(s1)!0.5!(s2)$) {};
        \node[anode] (a14) at ($(s1)!0.5!(s4)$) {};
        \node[anode] (a23) at ($(s2)!0.5!(s3)$) {};
        \node[anode] (a25) at ($(s2)!0.5!(s5)$) {};
        \node[anode] (a3) at ($(s3)!0.5!(s1)$) {};
        \node[anode] (a36) at ($(s3)!0.5!(s6)$) {};
        \node[anode] (a54) at ($(s5)!0.5!(s4)$) {};
        \path[font=\scriptsize]
        (s1) edge node[left]{$\alpha$} (a12)
        (a12) edge[->] (s2)
        (a12) edge[->] (s3)
        (s2) edge node[above, xshift=-1mm]{$\alpha$} (a23)
        (a23) edge[->] (s3)
        (s3) edge node[left]{$\alpha$} (a3)
        (a3) edge[->, bend left=40] (s3)
        (s1) edge node[above]{$\beta$} (a14)
        (a14) edge[->] (s4)
        (s2) edge node[below]{$\beta$} (a25)
        (a25) edge[->] (s5)
        (a25) edge[->, bend right=20] (s2)
        (s3) edge node[above]{$\beta$} (a36)
        (a36) edge[->, bend right=20] node[below left=-0.11 and -0.11, near start]{$\frac{2}{3}$} (s5)
        (a36) edge[->] node[above]{$\frac{1}{3}$} (s6)
        (s5) edge (a54)
        (a54) edge[->] node[right]{$\frac{2}{3}$} (s4)
        (a54) edge[->] node[below]{$\frac{1}{3}$} (s6)
        (s4) edge[->, out=30, in=60, loop] (s4)
        (s6) edge[->, out=30, in=60, loop] (s6)
        ;
    \end{tikzpicture}~
    \begin{tikzpicture}[mdp]
        \node[ps, cinit, init=left] (s1) {\nodebg{$s_1$}};
        \node[ps, cinit, below=1.2 of s1] (s2) {\nodebg{$s_2$}};
        \node[ps, cinit, below right=0.48 and 0.3 of s1] (s3) {\nodebg{$s_3$}};
        \node[ps, cinit, right=2.0 of s1] (s4) {\nodebg{$s_4$}};
        \node[ps, right=2.0 of s2, label={right:$\othergoal$}] (s5) {$s_5$};
        \node[ps, right=0.55 of s3, label={right:$\othergoal$}] (s6) {$s_6$};
        \node[anode, cinite] (a12) at ($(s1)!0.5!(s2)$) {};
        \node[anode, cinite] (a14) at ($(s1)!0.5!(s4)$) {};
        \node[anode, cinite] (a23) at ($(s2)!0.5!(s3)$) {};
        \node[anode] (a25) at ($(s2)!0.5!(s5)$) {};
        \node[anode, cinite] (a3) at ($(s3)!0.5!(s1)$) {};
        \node[anode] (a36) at ($(s3)!0.5!(s6)$) {};
        \path[font=\scriptsize]
        (s1) edge[cinite] node[left]{$\alpha$} (a12)
        (a12) edge[->, cinite] (s2)
        (a12) edge[->, cinite] (s3)
        (s2) edge[cinite] node[above, xshift=-1mm]{$\alpha$} (a23)
        (a23) edge[->, cinite] (s3)
        (s3) edge[cinite] node[left]{$\alpha$} (a3)
        (a3) edge[->, cinite, bend left=40] (s3)
        (s1) edge[cinite] node[above]{$\beta$} (a14)
        (a14) edge[->, cinite] (s4)
        (s2) edge node[below]{$\beta$} (a25)
        (a25) edge[->] (s5)
        (a25) edge[->, bend right=20] (s2)
        (s3) edge node[above]{$\beta$} (a36)
        (a36) edge[->, bend right=20] node[below left=-0.11 and -0.11, near start]{$\frac{2}{3}$} (s5)
        (a36) edge[->] node[above]{$\frac{1}{3}$} (s6)
        (s5) edge[->, out=30, in=60, loop] (s5)
        (s4) edge[->, cinite, out=30, in=60, loop] (s4)
        (s6) edge[->, out=30, in=60, loop] (s6)
        ;
    \end{tikzpicture}~
    \begin{tikzpicture}[mdp]
        \node[ps] (sb) {$s_\bot$};
        \node[ps, below=0.3 of sb, label={right:$\othergoal$}] (s5) {$s_5$};
        \node[ps, below=1.2 of sb, label={right:$\othergoal$}] (s6) {$s_6$};
        \node[ps, init=left, left=1 of s5] (si) {$\sinit$};
        \node[anode] (ai6) at ($(si)!0.5!(s6)$) {};
        \node[anode] (aib) at ($(si)!0.5!(sb)$) {};
        \path[font=\scriptsize]
        (s5) edge[->, out=30, in=60, loop] (s5)
        (s6) edge[->, out=30, in=60, loop] (s6)
        (sb) edge[->, out=30, in=60, loop] (sb)
        (si) edge node[above, xshift=-1.5mm]{$(s_2,\beta)$} (aib)
        (aib) edge[->] node[above]{$\frac{1}{2}$} (sb)
        (aib) edge[->] node[above=-.05, near end]{$\frac{1}{2}$} (s5)
        (si) edge node[below, xshift=-1.5mm]{$(s_3,\beta)$} (ai6)
        (ai6) edge[->] node[below=-.05, near end]{$\frac{2}{3}$} (s5)
        (ai6) edge[->] node[below]{$\frac{1}{3}$} (s6)
        ;
    \end{tikzpicture}
    \caption{Example MDPs $\mdpexb$ (left), $\mdp_{2,\circlearrowright}$ (middle), 
    $\tilde{\mdpexb}$ (right). The initial component is shown in blue, together with all their actions that are not exits. All unmarked transition probabilities are uniformly distributed over all transitions of an action.}
    \label{fig:exmtilde}
\end{figure}
\begin{example}\label{ex:step1a}
    Consider the example in \Cref{fig:exmtilde}, for which we solve an instance of \Cref{prob:threshold} with the property $\pr{\mdpexb}{\max}(\lozenge \set{s_4} \mid \lozenge \set{s_5, s_6}) \leq \nicefrac{1}{2}$. By \Cref{thm:threshold} we need to compute $\val(\nicefrac{1}{2})$.
    Thus, we consider the reward function $\rew^{\nicefrac{1}{2}}$ on $\mdpexb$. 
    By definition we must compute the rewards for transitions into $\cond=\set{s_5, s_6}$, these gain a reward of $\pr{s_5}{\max}(\lozenge\set{s_4})-\nicefrac{1}{2} = \nicefrac{2}{3}-\nicefrac{1}{2} = \nicefrac{1}{6}$, $\pr{s_6}{\max}(\lozenge\set{s_4})=0-\nicefrac{1}{2} = -\nicefrac{1}{2}$ respectively, and transitions into $\goal=\set{s_4}$: $\pr{s_4}{\max}(\lozenge\set{s_5, s_6})\cdot(1-\nicefrac{1}{2}) = 0\cdot\nicefrac{1}{2} = 0$.
    The memoryless policy $\sigma(s_1) = \alpha, \sigma(s_2) = \beta, \sigma(s_3) = \alpha$ has maximal reward $\ex{\mdpexb}{\sigma}(\rew^{\lambda=\nicefrac{1}{2}}_{\lozenge\set{s_4,s_5,s_6}}) = \nicefrac{1}{12}$. The policy $\sigma$ reaches $\cond$ with positive probability, thus $\sigma \in \pols[\mdpexb,\cond]$,
    and since $\ex{\mdpexb}{\sigma}(\rew^{\lambda=\nicefrac{1}{2}}_{\lozenge\set{s_4,s_5,s_6}}) > 0$, the property does not hold (\Cref{thm:threshold}).
\end{example}

\subsubsection*{Step 2a: Preprocessing $\mdp$.}
The example above does not yet discuss a method to find a policy $\sigma \in \pols[\mdp,\cond]$. We first simplify the expected total reward query by considering the MDP $\mdp_\circlearrowright$ obtained from $\mdp$ by making the states in $\goal \cup \cond$ absorbing, i.e., all outgoing transitions of each state $s \in \goal \cup \cond$ are replaced by a self-loop $s \xrightarrow{\bot} s$.
For ease of notation, reward functions $\rew$ for $\mdp$ also apply to $\mdp_\circlearrowright$, where $\rew(s,\bot,s) = 0$ for all $s \in \goal \cup \cond$.
On $\mdp_\circlearrowright$, we consider a new set of terminal states
\[
\othergoal ~\coloneq~ \cond \cup \bigset{ s \in \goal \bigmid \prsmax(\lozenge \cond) > 0 }.
\]
Intuitively, paths of $\mdp$ that reach a state in $\goal \setminus \othergoal$ (without visiting a state in $\cond$ before) have reward $0$ and do not reach $\cond$.
\begin{example}\label{ex:step1b}
    Consider $\mdp_{2,\circlearrowright}$ in \Cref{fig:exmtilde}. We have $\othergoal = \set{s_5, s_6}$. The goal state $s_4 \in \goal$ is not part of $\othergoal$, since it has no path to the evidence set~$\cond = \set{s_5, s_6}$. 
\end{example}
The following lemma establishes that $\val(\lambda)$ can be obtained by computing
(i) the non-zero rewards of $\rew^\lambda$ and (ii) an optimal expected total reward on $\mdp_\circlearrowright$.
The former
can be derived using two standard MDP model checking calls to obtain the probabilities $\prsmax(\lozenge \goal)$ for each $s \in \cond$ and $\prsmax(\lozenge \cond)$ for each $s \in \goal \setminus \cond$.
\begin{lemma}\label{lem:absorbing}
    $\displaystyle \val(\lambda) = \max_{\sigma \in \pols[{\mdp_\circlearrowright,\othergoal}]} \ex{\mdp_\circlearrowright}{\sigma} \big(\rew^\lambda_{\lozenge \othergoal}\big)$
\end{lemma}
\begin{proof}[sketch]
    Given a policy $\sigma \in \pols[\mdp,\cond]$, we obtain a policy $\sigma' \in \pols[\mdp_\circlearrowright, \othergoal]$ that mimics $\sigma$ until an absorbing state $s \in \goal \cup \cond$ is reached.
    Similarly, for a policy $\sigma' \in \pols[\mdp_\circlearrowright, \othergoal]$ a policy $\sigma \in \pols[\mdp,\cond]$ is constructed that mimics $\sigma'$ until a state in $\goal \cup \cond$ is reached from which point on the policy $\sigma$ mimics a policy that maximizes the probability to reach $\cond$.
    In both cases, we have
    $ \exms\big(\rew^\lambda_{\lozenge \goal \cup \cond}\big) = \ex{\mdp_\circlearrowright}{\sigma'} \big(\rew^\lambda_{\lozenge \othergoal}\big)$.
\end{proof}

\begin{remark}[Minimal Conditional Probabilities] \label{rem:minimal}
Our approach can be straightforwardly generalized to minimal conditional probabilities $\prmmin(\lozenge \goal \mid \lozenge \cond)$ except one edge case:
For the equivalent of \Cref{lem:absorbing}, we might have $\pols[\mdp_\circlearrowright, \othergoal] = \emptyset$ even if $\pols[\mdp,\cond] \neq \emptyset$. 
This case occurs if all paths that reach $\cond$ visit a state in $\goal \setminus \othergoal = \set{ s \in \goal \mid  \prsmin(\lozenge \cond) = 0}$ before. If this holds, we get $\prmmin(\lozenge \goal \mid \lozenge \cond) = 1$.
\end{remark}

\subsubsection*{Step 2b: Discard Policies with Undefined Value.}
Using our processing from Step 2a, we may assume an MDP $\mdp$ with terminal states $\othergoal \subseteq S$, all absorbing, and a reward function $\rew^\lambda \colon S \times \act \times S \to \RR$ where $\rew^\lambda(s,\alpha,s') \neq 0$ implies $s' \in \othergoal$.
We now provide an approach to decide whether $\max_{\sigma \in \pols[\mdp,\othergoal]} \exms(\rew^\lambda_{\lozenge \othergoal}) \sim 0$ for
${\sim} \in \set{<,\le, =, \ge,>}$.
Our procedure then allows to decide \Cref{prob:threshold} using \Cref{lem:absorbing,thm:threshold}.
The challenge is to limit the expected total reward analysis to policies $\sigma \in \pols[\mdp,\othergoal]$ that reach $\othergoal$ with non-zero probability.

For a policy $\sigma \in \pols[\mdp]$, the set 
$
\reach[\mdp](\sigma) ~=~ \bigset{ s \in S  \bigmid \prms(\lozenge \set{s}) > 0}
$ contains states that are reachable from the initial state $\sinit$ of $\mdp$ under $\sigma$ with non-zero probability. 
For $\sigma \in \pols[\mdp] \setminus \pols[\mdp,\othergoal ]$, the set $\othergoal$ is reached with probability 0, i.e., $\reach[\mdp](\sigma) \cap \othergoal = \emptyset$.
Furthermore, for each state $s \in \reach[\mdp](\sigma)$, there is a policy under which $\othergoal$ is reached from $s$ with probability 0, yielding
\[\cinit ~\coloneq~ 
\bigcup_{\sigma \in \pols[\mdp] \setminus \pols[\mdp,\othergoal]} \reach[\mdp](\sigma)
~\subseteq~ \bigset{ s \in S \bigmid \prsmin\big(\lozenge \othergoal) = 0}.
\]
We call the set $\cinit \subseteq S$ the \emph{initial component} of $\mdp$ and observe that strategies reaching $\othergoal$ must exit $\cinit$ eventually. We emphasise that $\cinit$ does not need to be a strongly connected component and is thus not a (traditional) end component.
\begin{example}\label{ex:step2b1}
    Consider $\mdp_{2,\circlearrowright}$ from \Cref{fig:exmtilde}. The initial component of $\mdp_{2,\circlearrowright}$ with states $\cinit = \{s_1, s_2, s_3, s_4\}$ is highlighted. It is not strongly connected. For each of its states there exists a policy which stays inside $\cinit$. For example for state $s_2$, the policy $\sigma(s_2)=\alpha,\sigma(s_3)=\alpha$ never reaches a state in $\othergoal = \set{s_5,s_6}$. 
\end{example}
\begin{lemma}\label{lem:exit_cinit}
    For all policies $\sigma \in \pols[\mdp]$:
    \[
    \prms(\lozenge \othergoal) > 0
\quad\text{iff}\quad
\sigma \in \pols[\othergoal,\mdp]
\quad\text{iff}\quad
\prms(\lozenge (S \setminus \cinit)) > 0. 
    \]
\end{lemma}
\begin{proof}
   The first equivalence is by definition of $\pols[\othergoal,\mdp]$. Moreover,
$\prms(\lozenge \othergoal) > 0$ implies 
$\prms(\lozenge (S \setminus \cinit)) > 0$ because $\cinit$ does not contain any state in $\othergoal$, i.e., $\othergoal \subseteq (S \setminus \cinit)$.
Finally, $\prms(\lozenge \othergoal) = 0$ implies 
$\prms(\lozenge (S \setminus \cinit)) = 0$ because
for $\sigma \notin \pols[\othergoal,\mdp]$ we have $\reach[\mdp](\sigma) \subseteq \cinit$ by definition of $\cinit$ and  $\prms(\lozenge \set {s}) = 0$ for all $s \in  (S \setminus \reach[\mdp](\sigma)) \supseteq (S \setminus \cinit)$.
\end{proof}
If $\cinit = \emptyset$, there is no policy $\sigma \in \pols[\mdp] \setminus \pols[\mdp,\othergoal]$, i.e. $\pols[\mdp,\othergoal] = \pols[\mdp]$.
We assume $\cinit \neq \emptyset$.
\Cref{lem:exit_cinit} yields that $\sigma \in \pols[\othergoal,\mdp]$ iff there is some state $s \in \cinit \cap \reach[\mdp](\sigma)$ via which the initial component is exited, i.e. $(s,\sigma(s)) \in \exits(\cinit)$ for
\[
\exits(\cinit)~\coloneq~\bigset{ (s,\alpha) \in S \times \act \bigmid s \in \cinit, \alpha \in \act(s), \text{ and } \supp(\Pbf(s,\alpha)) \nsubseteq \cinit}.
\]
We construct a new MDP $\tilde{\mdp}$ from $\mdp$ where the choice of an exiting state-action pair is explicit, while ensuring that the maximal expected reward is positive in $\mdp$ iff it is positive in $\tilde{\mdp}$.

\begin{definition}
\label{def:initialcomp}
The elimination of the initial component $\cinit \neq \emptyset$ of $\mdp$ yields the MDP $\tilde{\mdp}=(\tilde{S} {=} (S \setminus \cinit) \cup \set{\sinit, s_\bot}, \act \cup \exits(\cinit) \cup \set{\bot}, \tilde{\Pbf}, \sinit)$, where $\tilde{\Pbf}$ is defined as follows---using $\mu_{s,\alpha} \in \dist(\tilde{S})$ with $\mu_{s,\alpha}(s') = \Pbf(s,\alpha)(s')$ for $s' \in S \setminus \cinit$ and $\mu_{s,\alpha}(s_\bot) = \sum_{s' \in \cinit}\Pbf(s,\alpha,s')$:
\[
\tilde{\Pbf}(s,\alpha) ~=~
 \begin{cases}
\mu_{s,\alpha} & \text{if } s \in S \setminus \cinit, \alpha \in \act,\\
\mu_{s_e,\alpha_e} & \text{if } s = \sinit, \alpha = (s_e,\alpha_e) \in \exits(\cinit).\\
\set{s_\bot \mapsto 1} & \text{if } s = s_\bot, \alpha = \bot.
\end{cases}
 \]
\end{definition}
The reward function $\rew^\lambda$ for $\mdp$ is straightforwardly applied to $\tilde{\mdp}$. In particular, $\rew(s,\alpha,s')$ only depends on $s'$ and only assigns non-zero values for $s' \in S \setminus \cinit$.
\begin{example}
The initial component $\cinit$ of $\mdp_{2,\circlearrowright}$ from \Cref{fig:exmtilde} (middle) has two exits: $(s_2,\beta)$ and $(s_3, \beta)$. Eliminating $\cinit$ results in $\tilde{\mdpexb}$ in \Cref{fig:exmtilde} (right). 
The transition probabilities of $(s_2, \beta)$ are modified such that the probability of going back into the initial component is instead redirected to the state $s_\bot$.
\end{example}

\begin{lemma}
For ${\sim} \in \set{<,\le, =, \ge,>}$:
    \[
    \max_{\sigma \in \pols[\mdp,\othergoal]} \exms(\rew^\lambda_{\lozenge \othergoal}) ~\sim~ 0
    \quad \text{iff} \quad
    \max_{\tilde{\sigma} \in \pols[\tilde{\mdp}]}
    \ex{\tilde{\mdp}}{\tilde{\sigma}}(\rew^\lambda_{\lozenge \othergoal}) ~\sim~ 0.
    \]
\end{lemma}
\begin{proof}[sketch]
For the implication, let $\sigma \in \pols[\mdp,\othergoal]$ be a policy of $\mdp$ that maximises the expected reward $\exms(\rew^\lambda_{\lozenge \othergoal})$ among all policies in $\pols[\mdp,\othergoal]$.
We construct a policy $\tilde{\sigma}$ of $\tilde{\mdp}$ such that the comparisons of the induced expected rewards to zero coincide.
To see this, consider a case distinction on $\mathit{val}(s,\alpha)$ defined for $(s,\alpha) \in \exits(\cinit)$ as 
$\mathit{val}(s,\alpha) ~=~ \sum_{s' \in S \setminus \cinit}  \Pbf(s,\alpha,s') \cdot ( \ex{s'}{\max}(\rew^\lambda_{\lozenge T}) + \rew^\lambda(s,\alpha,s') )$:
\begin{itemize}
\item If $\mathit{val}(s,\alpha) > 0$ for some $(s,\alpha)\in \exits(\cinit)$, then the maximal expected reward in $\mdp$ is positive (as $\sigma$ is optimal, it will prefer staying inside the initial component $C^I$ and collect zero reward over picking an exit with negative value). Policy $\tilde{\sigma}$ for $\tilde{\mdp}$ mimics this by picking that exit in $s^I$.
\item If $\mathit{val}(s',\alpha') \leq 0$ for all $(s',\alpha')\in \exits(\cinit)$ and $\mathit{val}(s,\alpha)=0$ for some $(s,\alpha)\in \exits(\cinit)$, the maximal expected reward in $\mdp$ is 0 due to a similar argument and $\tilde{\sigma}$ for $\tilde{\mdp}$ can select an action at $s^I$ with zero value.
\item Otherwise, if $\mathit{val}(s,\alpha)<0$ for all $(s,\alpha) \in \exits(\cinit)$, all policies induce a negative value and the choice of $\tilde{\sigma}$ does not matter.
\end{itemize}
For the converse, any policy $\tilde{\sigma} \in \pols[\tilde{\mdp}]$ can be converted to a policy $\sigma \in \pols[\mdp,\othergoal]$ such that 
the comparisons of the induced expected rewards to zero coincide.
In particular, if $\tilde{\sigma}(\sinit) = (s,\alpha)$, then $\sigma$ mimics this by only exiting $\cinit$ via the exit $(s,\alpha)$. Paths of $\mdp$ that do not take that exit gain zero reward.
\end{proof}

\begin{example}\label{ex:step2b2}
Recall the reward function $\rew^{\lambda}$ for $\lambda=\nicefrac{1}{2}$ of $\mdp_{2,\circlearrowright}$ from \cref{ex:step1a}.
Applying $\rew^{\lambda}$ on $\tilde{\mdpexb}$ gives a reward of $\nicefrac{1}{6}$ to the transitions into $s_5$, and a reward of $-\nicefrac{1}{2}$ to the transitions into $s_6$.
The maximum expected reward is attained by the policy $\sigma(\sinit)=(s_2,\beta)$, which induces $\ex{\mdpexb}{\sigma}(\rew^{\lambda}_{\lozenge \set{s_5, s_6}}) = \nicefrac{1}{3}$. 
    While this value is not equal to the value found in \cref{ex:step1a}, it is on the same side of zero.
\end{example}

\subsubsection*{Efficiently Handling Acyclic MDPs.}
For acyclic MDPs, expected total rewards and reachability probabilities can be computed in linear time %
with dynamic programming~\cite{DBLP:books/wi/Puterman94}.
Our constructions preserve acyclicity.

\begin{proposition}\label{thm:acyclic}
For acyclic MDPs, \Cref{prob:threshold} can be solved in $\mathcal{O}(\sizeof{\mdp})$ time.
\end{proposition}

\subsubsection*{Policy Extraction.}
\label{sec:policy-extraction}
Applications---for example planning problems, colored MDPs (\Cref{sec:colored-mdps}), and our policy tracking optimization from \Cref{sec:valueproblem}---often require extracting an optimal policy on $\mdp$ which witnesses $\prmmax(\lozenge \goal \mid \lozenge \cond) \sim \lambda$, i.e., $\val(\lambda) \sim 0$.
In contrast to standard reachability probabilities,
memoryless policies are not sufficient for conditional reachability~\cite{DBLP:conf/tacas/AndresR08}. Inspired by~\cite{DBLP:conf/tacas/AndresR08,DBLP:conf/tacas/BaierKKM14}, we split the behavior of the policy into three memory states: i)~neither $\goal$ nor $\cond$ states were reached yet, ii) some state $s \in \goal$ was reached, and iii)~some state $s \in \cond$ was reached. Intuitively, in ii)~and iii)~we seek the optimal policy for query $\prmmax(\lozenge \cond)$ and  $\prmmax(\lozenge \goal)$, respectively. 
In case i), we  directly use $\tilde{\sigma}(s)$ obtained by solving $\tilde{\mdp}$ for states $s \notin \cinit$.
Finally, we construct the policy for states in $\cinit$. The action $\tilde{\sigma}(\sinit)$ defines the desired initial component exit, i.e., if $\tilde{\sigma}(\sinit) = (s_{e},\alpha_{e})$, then the constructed policy chooses $\alpha_{e}$ in state $s_{e}$. For the other states in $\cinit$ we compute a policy that reaches $s_{e}$ with non-zero probability while only taking actions that stay in~$\cinit$, which can be implemented by standard graph algorithms.

\begin{example}
    Continuing from \Cref{ex:step2b2}, the policy $\tilde{\sigma}(\sinit)=(s_2,\beta)$ induced the maximum value of $\val(\nicefrac{1}{2})$ in $\tilde{\mdpexb}$. The policy for $\mdpexb$ will have three memory states: $\sigma'\colon \set{m^\emptyset, m^{\cond}, m^{\goal}} \to S \to \act$. $\sigma'(m^{\cond})$ and $\sigma'(m^{\goal})$ (case ii and iii) take the only available action in the goal and evidence states. For $\sigma'(m^{\emptyset})$ (case i), we assign $\sigma'(m^{\emptyset})(s_2) = \beta$ as it is the chosen exit. The other initial component states get assigned the actions $\sigma'(m^{\emptyset})(s_1) = \alpha$ and $\sigma'(m^{\emptyset})(s_3) = \alpha$ since $\sigma'$ needs to reach $s_2$ and otherwise stay in $\cinit$ to avoid reaching an evidence state.
\end{example}

\section{Computing Optimal Conditional Reachability}
\label{sec:valueproblem}
This section concerns optimal conditional reachability probabilities:
\begin{problemstatement}[Conditional Reachability Optimization Problem]\label{prob:optimization}
\vspace{-1em}
Given an MDP $\mdp$ with sets of states $\goal,\cond \subseteq S$ and a precision $\varepsilon \in [0,1]$, compute a value $\lambda$ with $\lvert\lambda - \prmmax(\lozenge \goal \mid \lozenge \cond)\rvert \le \varepsilon$.
\end{problemstatement}

\begin{algorithm}[t]
\Input{MDP $\mdp$ with $\goal,\cond \subseteq S$, $\varepsilon \in [0,1]$}
\Output{$\lambda \in [0,1]$ with $\lvert\lambda - \prmmax(\lozenge \goal \mid \lozenge \cond)\rvert \le \varepsilon$}
$\ell \gets 0\,;~u \gets 1$ \tcp{Invariant: $\ell \le \prmmax(\lozenge \goal \mid \lozenge \cond) \le u$}
\While{$(u-\ell)/2 > \varepsilon$}{%
$\lambda \gets (\ell+u) / 2$ \label{alg:line:candidate}\tcp{current candidate}
\tcp{using $\val(\lambda) = \max_{\sigma \in \pols[\mdp,\cond]} \big(\prms(\lozenge \goal \cap \lozenge \cond) - \lambda \cdot \prms(\lozenge \cond)\big) $}%
\lIf(\tcp*[h]{$\lambda$ is a lower bound\label{line:upl}}){$\val(\lambda) \ge 0$}{$\ell \gets \lambda$}
\lIf(\tcp*[h]{$\lambda$ is an upper bound\label{line:upu}}){$\val(\lambda) \le 0$}{$u \gets \lambda$}
}
\Return $(\ell+u)/2$
    \caption{Bisection for Computing Optimal Conditional Probabilities}
    \label{alg:bisection}
\end{algorithm}

Our approach for \Cref{prob:optimization} embeds the idea of \Cref{prob:threshold} into a bisection algorithm outlined in \Cref{alg:bisection}.
The algorithm maintains a lower bound $\ell$ and an upper bound $u$ on the optimal conditional probability. In each iteration, we check $\prmmax(\lozenge \goal \mid \lozenge \cond) \sim \lambda$ for ${\sim} \in \set{\le, \ge}$ and the midpoint $\lambda = (\ell + u)/2$ using our techniques from \Cref{sec:threshold}.
The bounds $\ell$ and $u$ are updated accordingly, halving the width of the interval $[\ell, u]$ in each step.

\begin{proposition}
\Cref{alg:bisection} answers \Cref{prob:optimization} and for $\varepsilon > 0$ terminates within $\lceil\log_{2}(1/(2\cdot \varepsilon))\rceil$ iterations.
\end{proposition}

We discuss three modifications of the bisection algorithm.
The first handles exact computations, i.e., $\varepsilon = 0$. The latter two aim to accelerate convergence.
All three modifications can be combined.
\Cref{sec:experiments} empirically compares them. 

\subsubsection{Exact Computation.}
We exploit that $\prmmax(\lozenge \goal \mid \lozenge \cond)$ is a rational number to enable exact computations.
Specifically, we adapt the candidate selection in \Cref{alg:line:candidate} and pick the rational in the interval $[\ell, u]$ with minimal denominator instead.
The resulting procedure yields a traversal of the \emph{Stern-Brocot tree}~\cite{DBLP:books/aw/GKP1994} which allows us to find every rational in finitely many steps.
To ensure progress, our implementation decides heuristically in each iteration whether (i) a candidate close to the middle of the search interval $[\ell,u]$ is taken or (ii) whether a candidate in $[\ell,u]$ with minimal denominator is taken.

\subsubsection{Terminating Early Through Policy Tracking}\label{sec:poltracking}
For $\sigma \in \pols[\mdp,\cond]$, we write $\val^\sigma(\lambda) = \prms(\lozenge \goal \cap \lozenge \cond) - \lambda \cdot \prms(\lozenge \cond)$ as in \Cref{sec:threshold}.
We say that $\sigma$ induces $\val(\lambda)$ if $\val^\sigma(\lambda) = \val(\lambda)$. The next modification is based on the following observation: 
\begin{lemma}\label{lem:poltrack}
    Let $0 \le \ell \le u \le 1$.
    A policy $\sigma \in \pols[\mdp,\cond]$ that induces $\val(\ell)$ and $\val(u)$ also induces $\val(\lambda)$ for all $\lambda \in [\ell,u]$.
\end{lemma}
\begin{proof}
For any $\sigma' \in \pols[\mdp,\cond]$ we have
$\val^{\sigma'}(\ell) \le \val^\sigma(\ell)$
and 
$\val^{\sigma'}(u) \le \val^\sigma(u)$.
Since $\val^{\sigma'}(\lambda)$ and $\val^\sigma(\lambda)$ are linear in $\lambda$, this yields $\val^{\sigma'}(\lambda) \le \val^\sigma(\lambda)$ for all $\lambda \in [\ell,u]$.
\end{proof}
When updating the bounds $\ell$ and $u$ in \Cref{line:upl,line:upu}, our modification of \Cref{alg:bisection} keeps track of a policy $\sigma_\ell$ that induces $\val(\ell)$ and a policy $\sigma_u$ that induces $\val(u)$.
If both policies coincide, they induce the optimal conditional probability. In particular, \Cref{lem:poltrack} yields that $\sigma^* = \sigma_\ell = \sigma_u$ induces the optimal value $\val^\sigma(\lambda^*) = \val(\lambda^*)$ for $\lambda^* = \prmmax(\lozenge \goal \mid \lozenge \cond) \in [\ell,u]$.
With \Cref{thm:threshold}, we get $\val^{\sigma^*}(\lambda^*) = 0$, yielding  \[
\prmmax(\lozenge \goal \mid \lozenge \cond) ~=~ \lambda^* ~=~ \frac{\pr{\mdp}{\sigma^*}(\lozenge \goal \cap \lozenge \cond)}{\pr{\mdp}{\sigma^*}(\lozenge \cond)} ~=~ \pr{\mdp}{\sigma^*}(\lozenge \goal \mid \lozenge \cond)\,.\]
We finally evaluate the found policy $\sigma^*$---by computing the above fraction---to obtain the resulting conditional probability.
\begin{example}
We compute $\pr{\mdpexb}{\max}(\lozenge \set{s_4} \mid \lozenge \set{s_5, s_6})$ for MDP $\mdp$ from \Cref{fig:exmtilde} using policy tracking bisection. The bisection algorithm starts with the bounds $[0,1]$ and no upper or lower bound policies. The current candidate $\lambda$ becomes $\nicefrac{1}{2}$. As seen in \Cref{ex:step1b}, $\val(\nicefrac{1}{2}) \geq 0$, thus the lower bound becomes $\nicefrac{1}{2}$, and the lower bound policy is $\sigma_\ell(s_1) = \alpha, \sigma_\ell(s_2) = \beta, \sigma_\ell(s_3) = \alpha$. Next, the candidate becomes $\lambda=\nicefrac{3}{4}$ and $\val(\nicefrac{3}{4}) < 0$. The upper bound is updated to $\nicefrac{3}{4}$ and the upper bound policy turns out to be equal to the lower bound policy $\sigma_u=\sigma_\ell$. Thus, we can apply \Cref{lem:poltrack} with $\sigma^*=\sigma_\ell=\sigma_u$ and calculate  $\pr{\mdpexb}{\sigma^*}(\lozenge \set{s_4} \mid \lozenge \set{s_5, s_6}) = \frac{\nicefrac{1}{3}}{\nicefrac{1}{2}} = \nicefrac{2}{3}$.
\end{example}

\subsubsection{Advanced Bisection Bounds}
The following results allow us to establish tighter bounds on the conditional probability based on the values $\val(\lambda)$.

\begin{lemma}\label{lem:advanced}
Let $\lambda \in [0,1]$. If $\val(\lambda) \ge 0$ then
\[
\lambda + \nicefrac{\val(\lambda)}{\prmmax(\lozenge \cond)} ~\le~
\prmmax(\lozenge \goal \mid \lozenge \cond) ~\le~
\lambda + \nicefrac{\val(\lambda)}{\prmmin(\lozenge \cond)}\,.
\]
If $\val(\lambda) \le 0$ then
\[
\lambda + \nicefrac{\val(\lambda)}{\prmmin(\lozenge \cond)} ~\le~
\prmmax(\lozenge \goal \mid \lozenge \cond) ~\le~
\lambda + \nicefrac{\val(\lambda)}{\prmmax(\lozenge \cond)}\,.
\]
\end{lemma}
The proof is in \refarxiv{\cite{arxiv}}{\cref{proofs:valueproblem}}. When updating the bounds $\ell$ and $u$ in \Cref{line:upl,line:upu} of \Cref{alg:bisection}, sharper bounds based on the inequalities in \Cref{lem:advanced} may be chosen.
This modification requires pre-computing $\prmmin(\lozenge \cond)$ and $\prmmax(\lozenge \cond)$.

\section{Conditional Reachability in Colored MDPs}
\label{sec:colored-mdps}

Colored MDPs~\cite{DBLP:journals/jair/AndriushchenkoCMJK25,heck2025constrained} extend MDPs by grouping states. Intuitively, policies must take the same action across a group. This model captures various synthesis problems where additional constraints on policies are imposed.
Each policy $\sigma$ that satisfies these constraints induces an MC $\mdp^\sigma$. Colored MDPs thus compactly represent families of MCs, each induced by such a policy.

This section shows how the methods above can be effectively combined with an abstraction-refinement (AR) method~\cite{DBLP:journals/jair/AndriushchenkoCMJK25} to solve conditional reachability problem for families of MCs, represented by colored MDPs. In the experiments, we use this for probabilistic program verification and monitoring benchmarks.

For conciseness, we use a simplified definition of colored MDPs from~\cite{DBLP:conf/atva/MaasJ25}:
\begin{definition}[Colored MDPs]
A colored
MDP is a tuple $\mdp^C = (\mdp, C, c)$, where $\mdp$ is an MDP with states $S$, $C$ is a set of colors, and $c \colon S \rightarrow C$.
\end{definition}
\begin{definition}[Color-consistent policy]
A memoryless policy $\sigma$ for a colored MDP 
 is \emph{color consistent}, if for all states $s$, $s'$:  $c(s)=c(s')$ implies $\sigma(s) = \sigma(s')$.
\end{definition}
Policy synthesis for colored MDPs asks for a color-consistent policy such that the reachability probability exceeds a given threshold. This problem is NP-hard, but efficient heuristics based on abstraction-refinement, counterexample generalization, and SMT exist~\cite{DBLP:journals/jair/AndriushchenkoCMJK25,heck2025constrained}. We lift the problem to conditional reachability.

\begin{problemstatement}[Threshold Conditional Reachability for colored MDPs]\label{prob:cMDP}
\vspace{-1em}
Given a colored MDP $\mdp^C = (\mdp, C, c)$ with $\goal,\cond \subseteq S$ and $\lambda \in [0,1]$, decide whether there is a color-consistent policy $\sigma$ such that  $Pr_{\mdp}^{\sigma}(\lozenge \goal \mid \lozenge \cond) \ge \lambda$.
\end{problemstatement}

In the remainder of this section, we clarify the connection of colored MDPs to family of MCs and adapt the AR algorithm from~\cite{DBLP:journals/jair/AndriushchenkoCMJK25} to our~\cref{prob:cMDP}. %

\subsubsection*{Family of MCs.}
\label{sec:family-of-mcs}
Many synthesis problems consider MDPs under a restricted set of policies. The set of induced Markov chains is a \emph{family}:
\begin{definition}[Family of MCs]
    A family of MCs $\mathcal{F}$ given by colored MDP $\mdp^C$ is defined as $\mathcal{F} = \{\mdp^\sigma \mid \sigma \text{ is a colored consistent policy for } \mdp^C\}$.
\end{definition}
\begin{example}\label{ex:family-of-mcs}
    Consider the colored MDP $\mdp^C$ in Figure~\ref{fig:cmdp-examples} with underlying MDP from Figure~\ref{fig:exmtilde} and a coloring with $c(s_2)=c(s_3)$, so color-consistent policies must satisfy $\sigma(s_2)=\sigma(s_3)$. While the underlying MDP has eight policies, only four are color consistent, corresponding to the three MCs in Figure~\ref{fig:cmdp-examples}: ($M_1$) from ~$\sigma(s_1)=\alpha,~\sigma(s_2)=\sigma(s_3)=\beta$, ($M_2$) from $\sigma(s_1)=\sigma(s_2)=\sigma(s_2)=\alpha$, and ($M_3$) $\sigma(s_1)=\beta,~\sigma(s_2)=\sigma(s_3)=\alpha/\beta$ where $s_2$ and $s_3$ unreachable.
\end{example}

\begin{figure}[t]
    \centering
    \begin{tikzpicture}[mdp]
        \node[ps, cmdpblue, init=left] (s1) {$s_1$};
        \node[ps, cmdporange, below=1.2 of s1] (s2) {$s_2$};
        \node[ps, cmdporange, below right=0.48 and 0.3 of s1] (s3) {$s_3$};
        \node[ps, right=2.0 of s1, label={right:$\goal$}] (s4) {$s_4$};
        \node[ps, cmdpyellow, right=2.0 of s2, label={right:$\cond$}] (s5) {$s_5$};
        \node[ps, cmdpyellow, right=0.55 of s3, label={above:$\cond$}] (s6) {$s_6$};
        \node[anode] (a12) at ($(s1)!0.5!(s2)$) {};
        \node[anode] (a14) at ($(s1)!0.5!(s4)$) {};
        \node[anode] (a23) at ($(s2)!0.5!(s3)$) {};
        \node[anode] (a25) at ($(s2)!0.5!(s5)$) {};
        \node[anode] (a3) at ($(s3)!0.5!(s1)$) {};
        \node[anode] (a36) at ($(s3)!0.5!(s6)$) {};
        \node[anode] (a54) at ($(s5)!0.5!(s4)$) {};
        \path[font=\scriptsize]
        (s1) edge node[left]{$\alpha$} (a12)
        (a12) edge[->] (s2)
        (a12) edge[->] (s3)
        (s2) edge node[above, xshift=-1mm]{$\alpha$} (a23)
        (a23) edge[->] (s3)
        (s3) edge node[left]{$\alpha$} (a3)
        (a3) edge[->, bend left=40] (s3)
        (s1) edge node[above]{$\beta$} (a14)
        (a14) edge[->] (s4)
        (s2) edge node[below]{$\beta$} (a25)
        (a25) edge[->] (s5)
        (a25) edge[->, bend right=20] (s2)
        (s3) edge node[above]{$\beta$} (a36)
        (a36) edge[->, bend right=20] node[below left=-0.11 and -0.11, near start]{$\frac{2}{3}$} (s5)
        (a36) edge[->] node[above]{$\frac{1}{3}$} (s6)
        (s5) edge (a54)
        (a54) edge[->] node[right]{$\frac{2}{3}$} (s4)
        (a54) edge[->] node[below]{$\frac{1}{3}$} (s6)
        (s4) edge[->, out=30, in=60, loop] (s4)
        (s6) edge[->, out=30, in=60, loop] (s6)
        ;
    \end{tikzpicture}
    \begin{tikzpicture}[mdp]
        \node[] at (-0.35,-0.45) {$M_1$};
        \node[ps, init=left] (s1) {$s_1$};
        \node[ps, below=1.2 of s1] (s2) {$s_2$};
        \node[ps, below right=0.48 and 0.3 of s1] (s3) {$s_3$};
        \node[ps, right=2.0 of s1] (s4) {$s_4$};
        \node[ps, right=2.0 of s2] (s5) {$s_5$};
        \node[ps, right=0.55 of s3] (s6) {$s_6$};
        \path[font=\scriptsize]
        (s1) edge[->] (s2)
        (s1) edge[->] (s3)
        (s2) edge[->, out=30, in=60, loop] (s2)
        (s2) edge[->] (s5)
        (s3) edge[->] node[below=-.05]{$\frac{2}{3}$} (s5)
        (s3) edge[->] node[above]{$\frac{1}{3}$} (s6)
        (s4) edge[->, out=30, in=60, loop] (s4)
        (s5) edge[->] node[right=-.05]{$\frac{2}{3}$} (s4)
        (s5) edge[->] node[above=-.05]{$\frac{1}{3}$} (s6)
        (s6) edge[->, out=30, in=60, loop] (s6)
        ;
    \end{tikzpicture}
    \begin{tikzpicture}[mdp]
        \node[] at (-0.35,-0.45) {$M_2$};
        \node[ps, init=left] (s1) {$s_1$};
        \node[ps, below=1.2 of s1] (s2) {$s_2$};
        \node[ps, below right=0.48 and 0.3 of s1] (s3) {$s_3$};
        \path[font=\scriptsize]
        (s1) edge[->] (s2)
        (s1) edge[->] (s3)
        (s2) edge[->] (s3)
        (s3) edge[->, out=30, in=60, loop] (s3)
        ;
    \end{tikzpicture}
    \begin{tikzpicture}[mdp]
        \node[] at (-0.35,-0.45) {$M_3$};
        \node[ps, init=left] (s1) {$s_1$};
        \node[ps, below=1.2 of s1] (s4) {$s_4$};
        \path[font=\scriptsize]
        (s1) edge[->] (s4)
        (s4) edge[->, out=30, in=60, loop] (s4)
        ;
    \end{tikzpicture}
    \caption{(left) Colored MDP $\mdp^C$ with underlying MDP from Fig.\ref{fig:exmtilde}. The coloring restricts the policy space to policies where $\sigma(s_2) = \sigma(s_3)$ ($s_5$ and $s_6$ are irrelevant as only one action is available here). (right) three Markov chains in the family of MCs given by color consistent policies in $\mdp^C$.}
    \label{fig:cmdp-examples}
\end{figure}
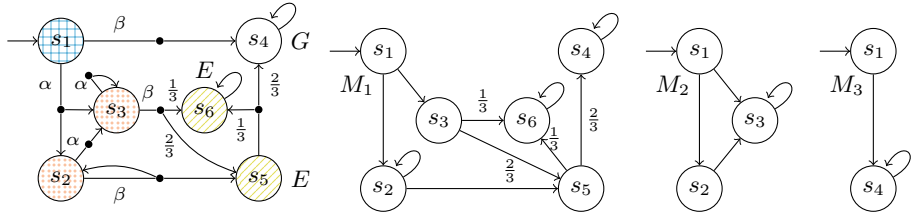

While an arbitrary family of MCs can be encoded as a colored MDP, the efficiency of this encoding depends on the similarity of the MCs in the family~\cite{DBLP:journals/jair/AndriushchenkoCMJK25}. 
Synthesis problems are not given by families of MCs, but immediately by colored MDPs, which are concise if the transition relation for most states coincides across family members.
Previous work shows use cases where families of over $10^{100}$ MCs are efficiently represented via colored MDPs~\cite{DBLP:journals/jair/AndriushchenkoCMJK25}.

\subsubsection*{Abstraction Refinement for Colored MDPs.}
\label{sec:ar-cmpds}

Consider a colored MDP $\mdp^C = (\mdp, C, c)$ and a constraint  $\ge \lambda$. 
Since the consistent policies in $\mdp^C$ are a subset of all policies, $\mdp$ provides a sound abstraction of $\mdp^C$. Therefore the policy $\sigma_{\max}$ maximising the conditional probability $Pr_{\mdp}(\lozenge \goal \mid \lozenge \cond)$ in the MDP $\mdp$ provides the upper bound $ub$ on maximal conditional probability achieved by any consistent policy in $\mdp^C$.
If $ub < \lambda$ then all policies, including the consistent policies violate $\ge \lambda$ and thus there is no consistent policy satisfying the constraint. If $ub \ge \lambda$ and $\sigma_{\max}$ is consistent, we have the solution of \cref{prob:cMDP}. In all other cases, we cannot draw any conclusion and the abstraction needs to be refined.

The goal of the refinement is to remove the inconsistent policy $\sigma_{\max}$ from the policies in the abstraction. This is achieved by splitting the set of policies in $\mdp$ into two (disjoint) subsets such that $\sigma_{\max}$ is not included in either subset. Let $\sigma_{\max}$ be inconsistent in the states $s$ and $s'$, i.e. $c(s)=c(s')$ while $\sigma(s) = \alpha$ and  $\sigma(s')=\beta$ (where $\alpha$ and $\beta$ are two different actions). We can efficiently represent the subsets by two (colored) sub-MDPs $\mdp_1$ and $\mdp_2$. %
In $\mdp_1$, we disable $\alpha$ in states $s$, $s'$ while in $\mdp_2$, we disable $\beta$. Any consistent policy from $\mdp$ exists in either $\mdp_1$ and $\mdp_2$. Hence, the AR leads to an iterative algorithm that analyses a sequence of colored MDPs $\mdp_i^C = (\mdp_i, C, c)$. In every iteration, either a satisfying consistent policy is found or $\mdp_i^C$ is either discarded (if $ub < \lambda$), or refined. The finite number of policies implies termination of the algorithm. 

\begin{example}\label{ex:cmdp-ar}
    Take the colored MDP $\mdp^C$ from Figure~\ref{fig:cmdp-examples} and a threshold $\geq 0.6 = \lambda$. We obtain $\prmmax(\lozenge G \mid \lozenge E) = \frac{2}{3}$ and policy $\sigma_{\max}$ with $\sigma_{\max}(s_1)=\sigma_{\max}(s_3)=\alpha$ and $\sigma_{\max}(s_2)=\beta$. Since the value exceeds $\lambda$, we check policy consistency. The policy is inconsistent: it selects different actions in $s_2$ and $s_3$. We refine the abstraction into two disjoint colored sub-MDPs: $\mdp_{1}^{C}$, $\alpha$ is enforced in $s_2$ and $s_3$, and $\mdp_{2}^{C}$, where $\alpha$ is not allowed in $s_2$, $s_3$. In $\mdp_{1}^{C}$ no policy reaches evidence states, so no policy can satisfy the threshold. For $\mdp_{2}^{C}$, we obtain $\pr{\mdp_{2}}{\max}(\lozenge G \mid \lozenge E)=\frac{5}{9} < \lambda$. Hence $\mdp_{2}^{C}$ is also discarded. Thus, no consistent policy satisfies $\lambda \geq 0.6$ in $\mdp^C$.
\end{example}

\subsubsection*{Refinement for Conditional Reachability.}
\label{sec:cmpd-advanced}
We integrate the approaches in \cref{sec:threshold,sec:valueproblem} with abstraction refinement~\cite{DBLP:journals/jair/AndriushchenkoCMJK25} and outline differences to the standard reachability case. First, for \cref{prob:cMDP}, the use of \cref{sec:threshold} suffices.
Second, the policy $\sigma_{\max}$  is typically inconsistent in several states, which raises a design choice for the refinement strategy. %
In~\cite{DBLP:journals/jair/AndriushchenkoCMJK25}, the heuristic is based on prioritizing states based on the expected number of visits multiplied by the probability of reaching the goal from that state, both under $\sigma_{\max}$. Conditional reachability queries do not provide these probabilities for every state~\cite{DBLP:conf/tacas/AndresR08}. Based on our preliminary experiments, we use a simple backward refinement that always refines the inconsistent state that has the maximal distance from $\sinit$.

\section{Experiments}
\label{sec:experiments}
\newcommand{\restart}{\textsf{restart}}
\newcommand{\treat}{\textsf{treat}}

We compare the \restart{} method~\cite{DBLP:conf/tacas/BaierKKM14} with our novel \textbf{t}otal \textbf{re}w\textbf{a}rd reduc\textbf{t}ion (\treat) based methods in~\cref{sec:threshold,sec:valueproblem}.
We answer the following research questions:
\begin{description}
    \item[(RQ1)] What is the fastest way to solve \Cref{prob:threshold} and \Cref{prob:optimization} (on MDPs)?
    \item[(RQ2)] Can the novel algorithms for the computation of conditional properties speed up solving of \cref{prob:cMDP} (on colored MDPs)?
    \item[(RQ3)] How does explicit support for conditional reachability in \storm{}~\cite{DBLP:journals/sttt/HenselJKQV22} affect runtime monitoring experiments as in \cite{DBLP:conf/cav/JungesTS20}?
\end{description}

\paragraph{Summary of the results.}
For (RQ1), the use of our new algorithms speeds up the computation on most benchmarks, for both floating-point and exact computations. The floating point results are accurate, whereas they are not always for the old restart method, and the performance has improved by orders of magnitude on some (acyclic) benchmarks.
For (RQ2), the results demonstrate not only the feasibility of adapting abstraction-refinement towards conditional probabilities, but also the importance of across-the-board performance improvements on small/easy MDPs in this setting for both floating-point and exact computations.
Beyond what is reported in (RQ1), the answers to (RQ3) highlight the impact of preprocessing, where already the old restart method with the novel preprocessing outperforms the previous `outside' construction of the restart MDP and its analysis by up to orders of magnitude.

\paragraph{Setup.} 
We run single-threaded experiments on an AMD Threadripper 5965wx, with a 10 minute timeout for RQ1 and RQ3 and a 15 minute timeout for RQ2. Memory consumption remains below 32 GB RAM. In the experiments, we report the model checking times, which exclude loading the model (from an explicit-state input) but include the necessary preprocessing. The supplementary material \cite{zenodo} contains all data points and the source code used to perform the experiments.

\subsection*{RQ1\quad Computing Conditional Properties on MDPs}
\label{sec:experiments:mdps}

\paragraph{Implementation.}
All methods (including \restart) are implemented within the model checker \storm{}~\cite{DBLP:journals/sttt/HenselJKQV22} and share the same process of model building and preprocessing, including the elimination of the initial component (\Cref{def:initialcomp})
All methods support exact and floating-point arithmetic and can be instantiated with any MDP solver. Here, we focus on the standard configurations in \storm\footnote{with exact guarantees using exact arithmetic and value iteration without hard guarantees for floating point arithmetic: like~\cite{DBLP:conf/tacas/HartmannsJQW23}, we observed that benchmarks where VI is wrong yield timeouts for optimistic VI.}. 

\paragraph{Benchmarks.} We use 318 benchmarks from four sources. 
(1)~We use (all) MDP benchmarks from \cite{DBLP:conf/tacas/BaierKKM14,DBLP:conf/sefm/MarckerB0K17}, which are adapted from standard benchmarks (\emph{consensus}, \emph{wlan}). 
(2)~From runtime monitoring for MDPs~\cite{DBLP:conf/cav/JungesTS20}, we collect challenging benchmarks\footnote{Which are not solved by all approaches within a $0.1$ seconds} using the same approach that was used to create benchmarks in~\cite{DBLP:conf/tacas/HartmannsJQW23} (\emph{refuel}, \emph{patrol}, \emph{evade}). 
(3)~We use MDPs from (uncertain/parametric) Bayesian networks~\cite{DBLP:journals/jair/SalmaniK23}, which correspond to the MDP semantics for the (interval) MCs~\cite{DBLP:journals/sttt/BadingsSSJ23} underlying these Bayesian networks (BNs) (\emph{alarm}, \emph{child}, and others).
(4)~Similarly, we add (interval-based) uncertainty to (all) MC benchmarks from~\cite{DBLP:conf/tacas/BaierKKM14,DBLP:conf/sefm/MarckerB0K17}, which are adapted from standard benchmarks (\emph{brp}, \emph{egl}, \emph{crowds}) and translate them to MDPs as above. Where applicable, we evaluate models on multiple properties. The runtime monitoring, BN and the \emph{brp} benchmark are acyclic.

\begin{figure}[t]
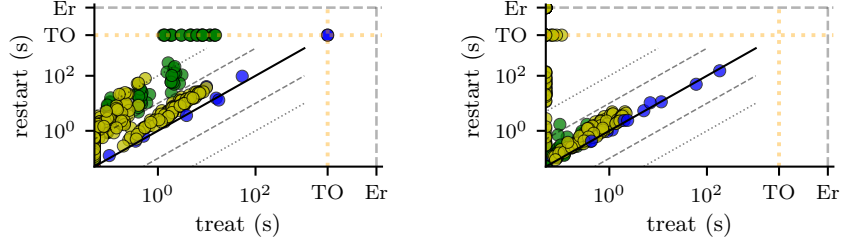

    \centering
    \begin{subfigure}[b]{0.48\textwidth}
    \resizebox{\linewidth}{!}{
\input{plots/mdps/scatter_bounded_bisection_exact_vs_restart_exact.pgf}
        }
        \label{fig:threshold:exact}
    \end{subfigure}
    \begin{subfigure}[b]{0.48\textwidth}
    \resizebox{\linewidth}{!}{
\input{plots/mdps/scatter_bounded_bisection_float_vs_restart_float.pgf}
}
        \label{fig:threshold:float}
    \end{subfigure}
    \vspace{-3mm}
    \caption{Runtime for solving \cref{prob:threshold}: \restart{} vs \treat{}. Exact arithmetic left, numerical computations right. Points refer to benchmark instances (from: \textcolor{red!60!black}{BNs}, \textcolor{blue!60!black}{MDPs}, \textcolor{green!60!black}{IMCs}, \textcolor{yellow!60!black}{RunMon}). Points above the diagonal correspond to a speedup of \treat{}, points above dashed lines to a 10x or 100x speedup.  }
    \label{fig:threshold}
\end{figure}
\subsubsection{Solving the Threshold Problem.}
We compare  \restart{} with \treat. In Fig.~\ref{fig:threshold} we provide a log-log scale scatter plot, where points correspond to benchmarks and provide a runtime comparison between two methods.

\paragraph{The exact case.}
We observe a considerable speed-up (${>}2\times$) on many models and a multi-order magnitude speed-up on some models. For a few benchmarks, the impact of the method choice is marginal as the initial component elimination alone yields tiny MDPs. The highest speed-up occurs on acyclic models with large initial components, where \treat{} uses dynamic programming and \restart{} requires repeated exact LU decompositions for the policy iteration. %

\paragraph{The numerical case.}
Numerical approaches relying on \restart{} do not (always) converge, as was observed before~\cite{DBLP:conf/sefm/MarckerB0K17,DBLP:conf/tacas/HartmannsJQW23}. Even when thresholding, this does lead to some errors: These errors are not due to unfortunate selections of thresholds, but due to values being completely off. For benchmarks where the numerical computations converge, we see significantly lower runtimes and a similar picture as in the exact case, although the effects are less pronounced.

\newcommand{\bisection}{\textsf{bis-std}}
\newcommand{\advbisection}{\textsf{bis-adv}}
\newcommand{\ptbisection}{\textsf{pt-bis-std}}
\newcommand{\ptadvbisection}{\textsf{pt-bis-adv}}

\begin{figure}[t]
    \centering
    \begin{subfigure}[b]{0.48\textwidth}
    \resizebox{\linewidth}{!}{
        \input{plots/mdps/scatter_quantitative_bisection-pt_exact_vs_restart_exact.pgf}
        }
        \label{fig:optimization:exact}
    \end{subfigure}
    \begin{subfigure}[b]{0.48\textwidth}
    \resizebox{\linewidth}{!}{
         \input{plots/mdps/scatter_quantitative_bisection-pt_float_vs_restart_float.pgf}
         }
        \label{fig:optimization:float}
    \end{subfigure}
    \vspace{-3mm}
    \caption{Runtime for solving \cref{prob:optimization}: \restart{} vs \ptbisection{}, analogous to \cref{fig:threshold}.}
    \label{fig:optimizationrestart}
\end{figure}

\subsubsection{Solving the Optimization Problem}
We compare  \restart{} with bisection-variants on top of \treat: We use standard (\bisection) and advanced bounds (\advbisection) and run it without and with policy tracking (\ptbisection, \ptadvbisection). 
In Fig.~\ref{fig:optimizationrestart}, we compare \ptbisection{} -- the empirically strongest configuration -- with \restart{}.  With exact arithmetic, results coincide and are used as a ground truth, whereas the numerical methods sometimes yield significant errors\footnote{We mark results that have a relative error $>10^{-3}$ as incorrect, unless ground truth and reported result are $<10^{-8}$. We use these liberal interpretation numbers to highlight the significance of the errors. }. The supplementary material contains additional scatter plots. 

\paragraph{The exact case.}
The configuration \ptbisection{} overall outperforms \restart{}, but not on every model: The \restart{} suffers significantly on acyclic models, but computing an exact and potentially exponentially large rational number with bisection can also be challenging. In fact, while \restart{} times out 60$\times$, \bisection{} and \advbisection{} 106$\times$ and 97$\times$, respectively. Policy tracking is effective, see also \cref{fig:optimization:eps}(left) and reduces time-outs to 9$\times$ for \ptbisection{} and 57$\times$ for \ptadvbisection.

\paragraph{The numerical case.}
The restart method does not reliably provide accurate results (as before) and provides wrong results on 171  benchmarks, whereas \ptbisection{} is always correct and solves many of them instantly. The overhead of any bisection method is limited as the number of iterations required to converge to a target precision of $10^{-6}$ is upper bounded to twenty for the most naive bisection strategy possible. Indeed, our bisection variations all yield lower number of iterations. Standard \bisection\ times out 13$\times$ compared to 6-7$\times$ for any improved variant.

\begin{figure}[t]
    \centering
    \begin{subfigure}[b]{0.48\textwidth}
        \resizebox{\linewidth}{!}{
        \input{plots/mdps/scatter_quantitative_bisection-pt_exact_vs_bisection_exact.pgf}
         \vspace{-3mm}
         }
    \end{subfigure}
    \begin{subfigure}[b]{0.48\textwidth}
        \resizebox{\linewidth}{!}{
         \input{plots/mdps/scatter_quantitative_bisection-pt_eps-exact_vs_restart_exact.pgf}
         }
          \vspace{-3mm}
    \end{subfigure}
    \caption{Runtime for solving \cref{prob:optimization}. Left: Exact arithmetic \bisection{} vs \ptbisection{}. Right: $\epsilon$-optimal \ptbisection{} vs. exact \restart{}. Legend as in \cref{fig:threshold}.}
     \label{fig:optimization:eps}
\end{figure}

\paragraph{Exact $\varepsilon$-guarantees.}
Due to the slow convergence of VI, solving reachability with exact arithmetic up to an $\varepsilon$-error is not effective. However, when using bisection-based methods, we can use \treat{} for computing exact results, but terminate once the exact result is approximated up to $10^{-6}$. As in the numerical case, this efficiently bounds the number of iterations, but provides strong guarantees. The resulting \cref{fig:optimization:eps}(right) clearly shows the advantages of the novel approach.

\subsection*{RQ2\quad Computing Conditional Properties on Colored MDPs}
\label{sec:experiments:colored-mdps}

\paragraph{Implementation and setup.}
We extended \tool{PAYNT}~\cite{DBLP:journals/jair/AndriushchenkoCMJK25}, a tool for solving colored MDPs, to support conditional reachability. \tool{PAYNT} uses \storm{} internally for model checking DTMCs and MDPs. Therefore, we can evaluate the  advantage of \treat{} versus \restart{} in the context of \cref{prob:cMDP}. We also evaluate the effect of only thresholding against computing the optimal value: In particular, we compare with \ptbisection\ that typically provides the best performance (see RQ1; results for the other bisection methods are presented in \refarxiv{\cite{arxiv}}{\cref{app:cmdp-extended-results}}). We report the overall runtime for \cref{prob:cMDP} and the abstraction-refinement iterations per second, for numerical and exact computations.

\begin{table}[t]
    \centering
        \caption{Performance comparison for color MDP problems. We report runtime in seconds and number of iterations per second. The left side reports results for numerical computation and the right side reports results for exact computation. The boldface values denote the fastest algorithm for each computation type.}
    \scalebox{0.9}{

\input{tables/colored_mdps_combined}

    }

    \label{tab:colored-mdps}
\end{table}

\paragraph{Benchmarks}
We use benchmarks from three different domains.
(1)~probabilistic program verification: the \emph{Caesar cipher (ceas)} with different input length and alphabet size and two variants of \emph{ladder network (ladder)}~\cite{DBLP:journals/pacmpl/HoltzenBM20} with varying possible inputs.
(2) Families of MCs from \tool{PAYNT}, but with conditional probability specifications (\emph{dpm~\cite{NPK+02}, virus}). 
(3)~Colored MDPs exported from monitor verification~\cite{DBLP:conf/atva/MaasJ25} where the problem is to verify whether a proposed monitor has no missed alarms. Table~\ref{tab:colored-mdps} reports the number of color-consistent policies $|\mdp^C|$ in the colored MDP $\mdp^C$ and the number of underlying states $\sizeof{\mdp}$.

\paragraph{Role of the refinement strategy.} Different polices computed by the individual algorithms can significantly affect the abstraction-refinement strategy (in terms of the number of iterations and the structure of the hardness of the consecutive queries) and thus the overall performance. %

\paragraph{Using \treat{} vs \restart.} For the numerical computation, the \restart \ method is unable to solve a majority of the benchmarks and \treat\ based methods are clearly superior. Except for the $airport$ benchmarks, this is due to the timing for individual queries. %
For $airport$, we observe the impact of the refinement strategy: despite a comparable number of MDP queries, \restart\ times out. For the exact computation, the advantage of \treat\ based methods is less prominent, but except for the smallest benchmarks 
\restart\ still lags considerably behind.

\paragraph{Using \treat\ vs. \ptbisection.} In most cases, \treat\ outperforms the bisection-based methods, since it reduces the number of value iterations needed to solve the given MDP query. However, the advantage of \treat\ can sometimes be mitigated by the impact of different policies on refinement.

\begin{table}[t]
    \centering
    
    \caption{Runtime monitoring examples: 3 methods, each evaluated on ten traces of 250 steps. We report the average and maximal time to evaluate a trace.}
    \input{tables/monitoring_table}
    \label{tab:monitoring}
\end{table}
\subsection*{RQ3 \quad Runtime Monitoring Revisited}
In this experiment, we evaluate the impact of support for conditional probabilities in MDPs into \storm{} in general and the use of \treat{} in particular, when considering an application in runtime monitoring. In particular, we replicate experiments from~\cite{DBLP:conf/cav/JungesTS20}. Somewhat simplified, that paper presents a method that takes a trace (of colors) $\tau$ to compare a series of conditional probabilities: 
\[ \prmmax(\lozenge \goal \mid \tau_{|1}) \leq \lambda, \dots, \prmmax(\lozenge G \mid \tau_{|n}) \leq \lambda,   \]
where $\Pr(\lozenge \goal \mid \tau_{|i})$ denotes the probability to reach states in $G$, given that the the states visited in the first $i$ steps match the coloring given in $\tau$. By unrolling, these queries can be reformulated into $\prmmax(\lozenge \goal \mid \lozenge \cond)$~\cite{DBLP:conf/tacas/BaierKKM14}.

\paragraph{Setup.}
The \textsf{previous} approach put forward in~\cite{DBLP:conf/cav/JungesTS20} unrolls the MDP immediately into the restart MDP and computes the conditional reachability probabilities using off-the-shelf MDP reachability in \storm{}. We adapted the code to only unroll and use the conditional reachability model checking. 
In contrast to RQ1, we did not filter for challenging individual instances, but run ten traces of 250 steps on models from the \cite{DBLP:conf/cav/JungesTS20}-repository.

\paragraph{Result.}
In \cref{tab:monitoring}, we observe not only the advantage of \treat{} vs \restart, but also how \storm{} provides efficient preprocessing that significantly improves the performance. The improvements in average time are particularly relevant for learning runtime monitors from data~\cite{DBLP:conf/atva/MaasJ25}, whereas a low maximal response time is particularly relevant when deploying a monitor~\cite{DBLP:series/lncs/BartocciDDFMNS18}.

\section{Conclusion}
Our novel approach to compute conditional reachability probabilities in MDPs clearly outperforms the previous state of the art. We show the benefits for various benchmarks from the literature, for Markov chain families, and in runtime monitoring applications. We would like to lift these approaches to symbolic state spaces and programs, as well as to interval MDPs or stochastic games. The extension towards Markov chain families raises additional questions how to generally handle specifications that require history-dependent policies. We furthermore hope to integrate support for conditional reachability  into~\cite{heck2025constrained}, which allows reasoning about families while directly supporting assumptions, e.g., on the reachability of $\cond$.

\paragraph{Acknowledgements.}
\inlinegraphics{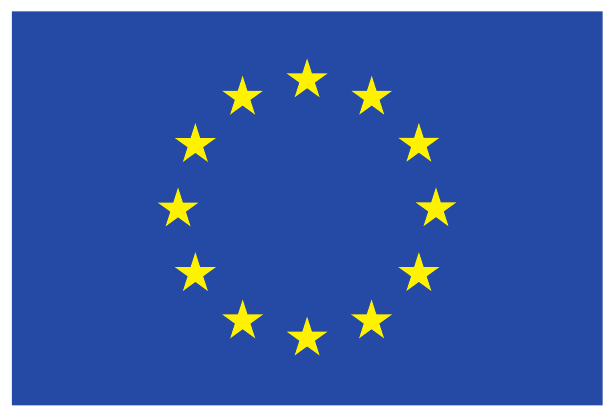} This work has been executed under the project VASSAL: ``Verification and Analysis for Safety and Security of Applications in Life'' funded by the European Union under Horizon Europe WIDERA Coordination and Support Action/Grant Agreement No. 10116002.
This work was partly supported by the NWO grant FuRoRe (OCENW.M.22.282) and by the KI-Starter Project Verifying AI Systems under Partial Observability of the Ministry of Culture and Science of the German State of North Rhine-Westphalia.

\paragraph{Competing Interests.}
The authors have no competing interests to declare that are relevant to the content of this article.

\paragraph{Data-Availability Statement.}
An artifact with source code and experimental data is available at \hyperlink{https://doi.org/10.5281/zenodo.19739061}{https://doi.org/10.5281/zenodo.19739061}~\cite{zenodo}. It contains the modifications done to \storm{}~\cite{DBLP:journals/sttt/HenselJKQV22}, \tool{PAYNT}~\cite{DBLP:journals/jair/AndriushchenkoCMJK25}, and \tool{premise}~\cite{DBLP:conf/cav/JungesTS20}. All algorithms implemented are part of \storm{} version 1.13.0 available at \hyperlink{https://stormchecker.org}{https://stormchecker.org}.

\clearpage
\bibliographystyle{splncs04}
\bibliography{literature}
\clearpage

\onlyarxiv{
\appendix

\section{Proofs}
\label{app:proofs}

\subsection{Proofs for \cref{sec:threshold}}
\label{proofs:threshold}
\subsubsection{Proof for \cref{thm:threshold}}

For $\sigma \in \pols[\mdp,\cond]$ we get
\begin{align*}
\prms(\lozenge \goal \mid \lozenge \cond) ~\sim~ \lambda
&\quad\textrm{iff}\quad
\frac{\prms(\lozenge \goal \cap \lozenge \cond)}{\prms(\lozenge \cond)} ~\sim~\lambda\\
&\quad\textrm{iff}\quad
\prms(\lozenge \goal \cap \lozenge \cond) - \lambda \cdot \prms(\lozenge \cond) ~\sim~0.
\end{align*}
We show the claim for ${\sim} = {\ge}$. Other cases are analogous.
\begin{align*}
\prmmax(\lozenge \goal \mid \lozenge \cond) ~\ge~ \lambda
&\quad\textrm{iff}\quad
\exists \sigma \in \pols[\mdp,\cond]\,\colon~\prms(\lozenge \goal \mid \lozenge \cond) ~\ge~ \lambda\\
&\quad\textrm{iff}\quad
\exists \sigma \in \pols[\mdp,\cond]\,\colon~
\prms(\lozenge \goal \cap \lozenge \cond) - \lambda \cdot \prms(\lozenge \cond) ~\ge~0\\
&\quad\textrm{iff}\quad
\max_{\sigma \in \pols[\mdp,\cond]} \Big(\prms(\lozenge \goal \cap \lozenge \cond) - \lambda \cdot \prms(\lozenge \cond)\Big) ~\ge~0.
\end{align*}\qed

\subsection{Proofs for \cref{sec:valueproblem}}
\label{proofs:valueproblem}
\subsubsection{Proof for \cref{lem:advanced}}
Let $\sigma^\lambda \in \pols[\mdp,\cond]$ be a policy that induces $\val(\lambda)$. We have:
\begin{align*}
\lambda + \frac{\val(\lambda)}{\pr{\mdp}{\sigma^\lambda}(\lozenge \cond)}
~&=~\lambda + \frac{\pr{\mdp}{\sigma^\lambda}(\lozenge \goal \cap \lozenge \cond) - \lambda \cdot \pr{\mdp}{\sigma^\lambda}(\lozenge \cond)}{\pr{\mdp}{\sigma^\lambda}(\lozenge \cond)}
~\le~ \prmmax(\lozenge \goal \mid \lozenge \cond).
\end{align*}
Furthermore, let $\sigma^* \in \pols[\mdp,\cond]$ be a policy that induces the maximal conditional reachability probability. We have:
\begin{align*}
    \prmmax(\lozenge \goal \mid \lozenge \cond)
    ~&=~ \lambda + \frac{\pr{\mdp}{\sigma^*}(\lozenge \goal \cap \lozenge \cond) - \lambda \cdot \pr{\mdp}{\sigma^*}(\lozenge \cond)}{\pr{\mdp}{\sigma^*}(\lozenge \cond)}
    ~\le~\lambda + \frac{\val(\lambda)}{\pr{\mdp}{\sigma^*}(\lozenge \cond)}.
\end{align*}
Using these inequalities, the lemma follow immediately as
\[
\prmmin(\lozenge \cond) ~\le~ \pr{\mdp}{\sigma^\lambda}(\lozenge \cond)\,,~\pr{\mdp}{\sigma^*}(\lozenge \cond) ~\le~ \prmmax(\lozenge \cond).
\]

\section{Experimental Data}
\label{appendix:experiments}

\subsection{Additional Scatter Plots for \cref{prob:optimization}}
\label{app:scatter-mdp}

We provide scatter plots comparing \ptbisection{} against other bisection variants across all arithmetic modes; the legend is as in \cref{fig:threshold}.

\begin{figure}[H]
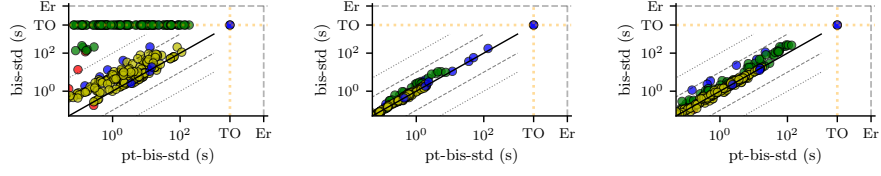

    \centering
    \begin{subfigure}[b]{0.32\textwidth}
        \resizebox{\linewidth}{!}{\input{plots/mdps/scatter_quantitative_bisection-pt_exact_vs_bisection_exact.pgf}}
    \end{subfigure}
    \begin{subfigure}[b]{0.32\textwidth}
        \resizebox{\linewidth}{!}{\input{plots/mdps/scatter_quantitative_bisection-pt_float_vs_bisection_float.pgf}}
    \end{subfigure}
    \begin{subfigure}[b]{0.32\textwidth}
        \resizebox{\linewidth}{!}{\input{plots/mdps/scatter_quantitative_bisection-pt_eps-exact_vs_bisection_eps-exact.pgf}}
    \end{subfigure}
    \caption{\ptbisection{} vs \bisection{}: exact (left), numerical (middle), $\varepsilon$-exact (right).}
    \label{fig:app:bisection}
\end{figure}

\begin{figure}[H]
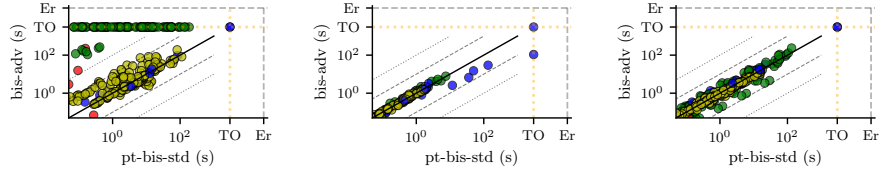

    \centering
    \begin{subfigure}[b]{0.32\textwidth}
        \resizebox{\linewidth}{!}{\input{plots/mdps/scatter_quantitative_bisection-pt_exact_vs_bisection-advanced_exact.pgf}}
    \end{subfigure}
    \begin{subfigure}[b]{0.32\textwidth}
        \resizebox{\linewidth}{!}{\input{plots/mdps/scatter_quantitative_bisection-pt_float_vs_bisection-advanced_float.pgf}}
    \end{subfigure}
    \begin{subfigure}[b]{0.32\textwidth}
        \resizebox{\linewidth}{!}{\input{plots/mdps/scatter_quantitative_bisection-pt_eps-exact_vs_bisection-advanced_eps-exact.pgf}}
    \end{subfigure}
    \caption{\ptbisection{} vs \advbisection{}: exact (left), numerical (middle), $\varepsilon$-exact (right).}
    \label{fig:app:bisection-adv}
\end{figure}

\begin{figure}[H]
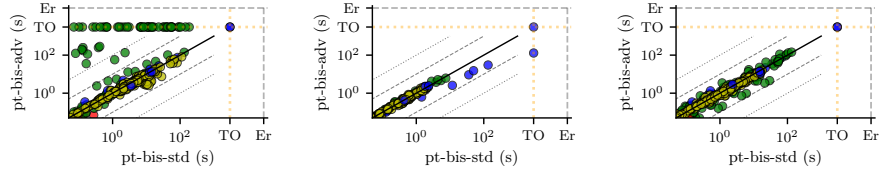

    \centering
    \begin{subfigure}[b]{0.32\textwidth}
        \resizebox{\linewidth}{!}{\input{plots/mdps/scatter_quantitative_bisection-pt_exact_vs_bisection-advanced-pt_exact.pgf}}
    \end{subfigure}
    \begin{subfigure}[b]{0.32\textwidth}
        \resizebox{\linewidth}{!}{\input{plots/mdps/scatter_quantitative_bisection-pt_float_vs_bisection-advanced-pt_float.pgf}}
    \end{subfigure}
    \begin{subfigure}[b]{0.32\textwidth}
        \resizebox{\linewidth}{!}{\input{plots/mdps/scatter_quantitative_bisection-pt_eps-exact_vs_bisection-advanced-pt_eps-exact.pgf}}
    \end{subfigure}
    \caption{\ptbisection{} vs \ptadvbisection{}: exact (left), numerical (middle), $\varepsilon$-exact (right).}
    \label{fig:app:bisection-adv-pt}
\end{figure}

\begin{figure}[H]
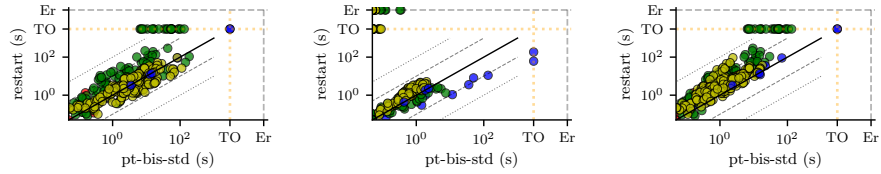

    \centering
    \begin{subfigure}[b]{0.32\textwidth}
        \resizebox{\linewidth}{!}{\input{plots/mdps/scatter_quantitative_bisection-pt_exact_vs_restart_exact.pgf}}
    \end{subfigure}
    \begin{subfigure}[b]{0.32\textwidth}
        \resizebox{\linewidth}{!}{\input{plots/mdps/scatter_quantitative_bisection-pt_float_vs_restart_float.pgf}}
    \end{subfigure}
    \begin{subfigure}[b]{0.32\textwidth}
        \resizebox{\linewidth}{!}{\input{plots/mdps/scatter_quantitative_bisection-pt_eps-exact_vs_restart_exact.pgf}}
    \end{subfigure}
    \caption{\ptbisection{} vs \restart{}: exact (left), numerical (middle), $\varepsilon$-exact (right).}
    \label{fig:app:restart}
\end{figure}

\subsection{Runtimes for \cref{prob:cMDP}}
\label{app:cmdp-extended-results}

In this section, we provide two tables comparing the impact of our novel algorithms on the \cref{prob:cMDP} benchmarks.

\begin{table}[H]
    \centering
    \caption{Results for all \treat{} based methods for the numerical case of \cref{prob:cMDP}.}
    \scalebox{0.8}{
    \input{tables/colored_mdps_extended_table}
    }
    \label{tab:appendix-cmpds-extended}
\end{table}

\begin{table}[H]
    \centering
    \caption{Results for all \treat{} based methods for the exact case of \cref{prob:cMDP}.}
    \scalebox{0.8}{
    \input{tables/colored_mdps_exact_extended_table}
    }
    \label{tab:appendix-cmpds-exact-extended}
\end{table}

}

\end{document}

%% file: tables/colored_mdps_combined.tex
\begin{tabular}{l@{\hskip 4pt}r@{\hskip 4pt}r@{\hskip 8pt}r@{\hskip 4pt}l@{\hskip 8pt}r@{\hskip 4pt}l@{\hskip 8pt}r@{\hskip 4pt}l@{\hskip 8pt}r@{\hskip 4pt}l@{\hskip 8pt}r@{\hskip 4pt}l@{\hskip 8pt}r@{\hskip 4pt}l@{\hskip 8pt}}
\toprule
 &  &  & \multicolumn{6}{c@{\hskip 8pt}}{numerical} & \multicolumn{6}{c@{\hskip 8pt}}{exact} \\
\cmidrule(lr{1.5em}){4-9}\cmidrule(lr{1.5em}){10-15}
\multirow{2}{*}{Model} & \multirow{2}{*}{$|\mdp^C|$} & \multirow{2}{*}{$|S|$} & \multicolumn{2}{c@{\hskip 8pt}}{\treat} & \multicolumn{2}{c@{\hskip 8pt}}{\ptbisection} & \multicolumn{2}{c@{\hskip 8pt}}{\restart{}} & \multicolumn{2}{c@{\hskip 8pt}}{\treat} & \multicolumn{2}{c@{\hskip 8pt}}{\ptbisection} & \multicolumn{2}{c@{\hskip 8pt}}{\restart{}} \\
\cmidrule(lr{1.5em}){4-5}\cmidrule(lr{1.5em}){6-7}\cmidrule(lr{1.5em}){8-9}\cmidrule(lr{1.5em}){10-11}\cmidrule(lr{1.5em}){12-13}\cmidrule(lr{1em}){14-15}
 & & & time & it/s & time & it/s & time & it/s & time & it/s & time & it/s & time & it/s \\
\midrule
ceas-4-10 & 1e6 & 265 & \textbf{<1s} & 1978 & \textbf{<1s} & 1978 & 120s & 1 & \textbf{<1s} & 848 & \textbf{<1s} & 712 & \textbf{<1s} & 574 \\
ceas-10-8 & 1e8 & 451 & \textbf{29s} & 1500 & 30s & 1459 & TO & <1 & \textbf{66s} & 654 & 71s & 610 & 86s & 507 \\
ladder & 8.6e11 & 251 & TO & 1204 & TO & 1157 & TO & 917 & TO & 411 & TO & 356 & TO & 196 \\
ladder-in & 8.6e12 & 503 & \textbf{305s} & 517 & 315s & 532 & TO & <1 & \textbf{742s} & 225 & TO & 163 & TO & 130 \\
dpm & 1620.0 & 1976 & \textbf{4s} & 75 & 46s & 6 & TO & <1 & \textbf{9s} & 29 & TO & <1 & 14s & 19 \\
dpm-q & 1.6e4 & 1733 & \textbf{5s} & 58 & 49s & 6 & TO & <1 & \textbf{10s} & 28 & TO & <1 & 15s & 19 \\
virus & 1.2e4 & 118 & \textbf{2s} & 1312 & TO & <1 & 5s & 659 & 12s & 344 & 205s & 19 & \textbf{10s} & 414 \\
airport & 1.3e12 & 51 & \textbf{58s} & 4726 & 94s & 4754 & TO & 3329 & 204s & 1339 & \textbf{194s} & 2299 & 631s & 2142 \\
airport-b & 1.8e14 & 115 & \textbf{21s} & 1583 & 29s & 1610 & TO & 565 & \textbf{39s} & 852 & 60s & 786 & 361s & 672 \\
\bottomrule
\end{tabular}

%% file: tables/monitoring_table.tex
\begin{tabular}{lrrrrrr}
\toprule
\multirow{2}{*}{Model} & \multicolumn{2}{c}{treat} & \multicolumn{2}{c}{restart} & \multicolumn{2}{c}{previous} \\
    \cmidrule(lr){2-3} \cmidrule(lr){4-5} \cmidrule(lr){6-7}
 & $t_\text{avg}$ & $t_\text{max}$ & $t_\text{avg}$ & $t_\text{max}$ & $t_\text{avg}$ & $t_\text{max}$ \\
\midrule
airportB-3-300-30 & \textbf{0.42} & \textbf{0.54} & 6.71 & 12.2 & 21.7 & 23.9 \\
evadeI-10 & \textbf{0.10} & \textbf{0.12} & 0.35 & 0.44 & 1.37 & 2.67 \\
evadeV-9-3 & \textbf{0.86} & \textbf{1.71} & 6.62 & 24.8 & 17.6 & 54.1 \\
refuelA-50-80 & \textbf{0.29} & \textbf{0.38} & 1.00 & 1.39 & 3.38 & 6.13 \\
\bottomrule
\end{tabular}

%% file: tables/colored_mdps_extended_table.tex
\begin{tabular}{lrr@{\hskip 12pt}r@{\hskip 4pt}l@{\hskip 12pt}r@{\hskip 4pt}l@{\hskip 12pt}r@{\hskip 4pt}l@{\hskip 12pt}r@{\hskip 4pt}l@{\hskip 12pt}r@{\hskip 4pt}l@{\hskip 12pt}r@{\hskip 4pt}l@{\hskip 12pt}}
\toprule
Model & $|\mathcal{F}|$ & $|M|$ & \multicolumn{2}{c@{\hskip 12pt}}{treat} & \multicolumn{2}{c@{\hskip 12pt}}{bis-std} & \multicolumn{2}{c@{\hskip 12pt}}{bis-adv} & \multicolumn{2}{c@{\hskip 12pt}}{pt-bis-std} & \multicolumn{2}{c@{\hskip 12pt}}{pt-bis-adv} & \multicolumn{2}{c@{\hskip 12pt}}{Restart} \\
\midrule
ceas-4-10 & 1e6 & 265 & \textbf{<1s} & 1978 & <1s & 1978 & <1s & 1978 & <1s & 1978 & <1s & 1978 & 120s & 1 \\
ceas-10-8 & 1e8 & 451 & \textbf{29s} & 1500 & 29s & 1486 & 30s & 1472 & 30s & 1459 & 30s & 1470 & TO & 0 \\
ladder & 8.6e11 & 251 & TO & 1204 & TO & 1222 & TO & 1265 & TO & 1157 & TO & 1171 & TO & 917 \\
ladder-in & 8.6e12 & 503 & \textbf{305s} & 517 & TO & 670 & TO & 682 & 315s & 532 & 358s & 530 & TO & 0 \\
dpm & 1620.0 & 1976 & \textbf{4s} & 75 & 49s & 5 & 44s & 6 & 46s & 6 & 44s & 6 & TO & 0 \\
dpm-q & 1.6e4 & 1733 & \textbf{5s} & 58 & 52s & 5 & 46s & 6 & 49s & 6 & 46s & 6 & TO & 0 \\
virus & 1.2e4 & 118 & \textbf{2s} & 1312 & TO & 0 & TO & 0 & TO & 0 & TO & 0 & 5s & 659 \\
airport & 1.3e12 & 51 & 58s & 4726 & \textbf{58s} & 4740 & 59s & 4634 & 94s & 4754 & 95s & 4689 & TO & 3329 \\
airport-b & 1.8e14 & 115 & \textbf{21s} & 1583 & 31s & 1648 & 32s & 1609 & 29s & 1610 & 29s & 1627 & TO & 565 \\
\bottomrule
\end{tabular}

%% file: tables/colored_mdps_exact_extended_table.tex
\begin{tabular}{lrr@{\hskip 12pt}r@{\hskip 4pt}l@{\hskip 12pt}r@{\hskip 4pt}l@{\hskip 12pt}r@{\hskip 4pt}l@{\hskip 12pt}r@{\hskip 4pt}l@{\hskip 12pt}r@{\hskip 4pt}l@{\hskip 12pt}r@{\hskip 4pt}l@{\hskip 12pt}}
\toprule
Model & $|\mathcal{F}|$ & $|M|$ & \multicolumn{2}{c@{\hskip 12pt}}{treat} & \multicolumn{2}{c@{\hskip 12pt}}{bis-std} & \multicolumn{2}{c@{\hskip 12pt}}{bis-adv} & \multicolumn{2}{c@{\hskip 12pt}}{pt-bis-std} & \multicolumn{2}{c@{\hskip 12pt}}{pt-bis-adv} & \multicolumn{2}{c@{\hskip 12pt}}{Restart} \\
\midrule
ceas-4-10 & 1e6 & 265 & \textbf{<1s} & 848 & <1s & 593 & <1s & 742 & <1s & 712 & <1s & 712 & <1s & 574 \\
ceas-10-8 & 1e8 & 451 & \textbf{66s} & 654 & 127s & 342 & 79s & 551 & 71s & 610 & 80s & 542 & 86s & 507 \\
ladder & 8.6e11 & 251 & TO & 411 & TO & 32 & TO & 85 & TO & 356 & TO & 142 & TO & 196 \\
ladder-in & 8.6e12 & 519 & \textbf{742s} & 225 & TO & 48 & TO & 82 & TO & 163 & TO & 59 & TO & 130 \\
dpm & 1620.0 & 1723 & \textbf{9s} & 29 & TO & 0 & TO & 0 & TO & 0 & 649s & 0 & 14s & 19 \\
dpm-q & 1.6e4 & 1733 & \textbf{10s} & 28 & TO & 0 & TO & 0 & TO & 0 & 614s & 0 & 15s & 19 \\
virus & 1.2e4 & 118 & 12s & 344 & 197s & 20 & 99s & 40 & 205s & 19 & 101s & 40 & \textbf{10s} & 414 \\
airport & 1.3e12 & 54 & 204s & 1339 & 199s & 1372 & \textbf{181s} & 1513 & 194s & 2299 & 263s & 1691 & 631s & 2142 \\
airport-b & 1.8e14 & 115 & \textbf{39s} & 852 & 112s & 434 & 123s & 396 & 60s & 786 & 105s & 446 & 361s & 672 \\
\bottomrule
\end{tabular}